\documentclass[sigconf]{acmart}
\usepackage{balance}

\newtheorem{hypothesis}{Hypothesis}

\newcommand{\dom}{\textup{dom}}
\AtBeginDocument{%
  \providecommand\BibTeX{{%
    \normalfont B\kern-0.5em{\scshape i\kern-0.25em b}\kern-0.8em\TeX}}}

\copyrightyear{2021}
\acmYear{2021}
\setcopyright{acmlicensed}
\acmConference[PODS '21] {Proceedings of the 40th ACM SIGMOD-SIGACT-SIGAI Symposium on Principles of Database Systems}{June 20--25, 2021}{Virtual Event, China}
\acmBooktitle{Proceedings of the 40th ACM SIGMOD-SIGACT-SIGAI Symposium on Principles of Database Systems (PODS '21), June 20--25, 2021, Virtual Event, China}
\acmPrice{15.00}
\acmISBN{978-1-4503-8381-3/21/06}
\acmDOI{10.1145/3452021.3458814}

\begin{document}
\fancyhead{}




\title{Modern Lower Bound Techniques in Database Theory and Constraint Satisfaction}

\author{D\'aniel Marx}
\email{marx@cispa.de}
\orcid{0000-0002-5686-8314}
\affiliation{%
  \institution{CISPA Helmholtz Center for Information Security}
  \streetaddress{Saarland Informatics Campus}
  \city{Saarbrücken}
  \country{Germany}
}


\begin{abstract}
  Conditional lower bounds based on $\textup{P}\neq \textup{NP}$, the Exponential-Time Hypothesis (ETH), or similar complexity assumptions can provide very useful information about what type of algorithms are likely to be possible. Ideally, such lower bounds would be able to demonstrate that the best known algorithms are essentially optimal and cannot be improved further. In this tutorial, we overview different types of lower bounds, and see how they can be applied to problems in database theory and constraint satisfaction.

\end{abstract}




\maketitle

\section{Introduction}

The design of efficient algorithms is in the focus of a large part of theoretical computer science research. The practical need to solve computational problems efficiently makes the systematic study of algorithmic efficiency highly motivated. Decades of research in algorithm design discovered mathematically beautiful and sometimes very practical algorithmic techniques that gave us deep insights into efficient computation in a wide range of contexts and application domains.  The field of computational complexity treats algorithmic problems and computation as formal mathematical objects and tries to prove relationships between them \cite{PapadimitriouBook,DBLP:books/daglib/0023084}.

Given the abundance of computation in our modern word, it is justified to consider algorithms and computation as a fundamental mathematical objects, on par with  basic objects in geometry, algebra, and combinatorics. Researchers in computational complexity try to learn as much as possible about the mathematical nature of computation. But more pragmatically, computational complexity can give very important messages to algorithm designers. By giving information about limits of computation, it can prevent researchers from wasting time in dead ends of study: trying to design algorithm for problems that cannot be efficiently solved.

Having techniques to prove negative results can profoundly change the way research in algorithms is done. For example, the theory of NP-hardness changed the search for polynomial-time algorithms from a hit and miss effort to a more systematically doable project. Without NP-hardness, we would not be able to distinguish problems that do not admit polynomial-time algorithms from problems where we just were not yet successful in finding algorithms. But with the possibility of giving negative evidence in the form of NP-hardness, the lack of a known answer means that the question is still an active research problem: we typically expect that the algorithmic problem at hand can be eventually classified as either polynomial-time solvable or NP-hard, and it is worth trying to resolve the question one way or the other.

In a sense, computation complexity  has not progressed much in the past 50 years despite intense efforts: the core questions underlying the hardness of computation, such as the celebrated $\textup{P}\neq \textup{NP}$ problem, are still wide open. However, by accepting certain well-chosen complexity assumptions, such as the $\textup{P}\neq \textup{NP}$ hypothesis, we can obtain \emph{conditional lower bounds} explaining the apparent complexity of a large number of problems. As a general theme in computational complexity research, we can see a proliferation of new assumptions. These assumptions typical postulate that a certain type of algorithm does not exist for a particular fundamental problem (e.g., for Boolean satisfiability). The assumptions are chosen to be both plausible and have strong explanatory power: in many cases, they are able to show that the algorithms that we currently have are optimal and cannot be improved any further. Indeed, from the viewpoint of algorithm design, this is precisely the role of computation complexity: to separate problems where our current knowledge is complete from problems where there are still algorithmic ideas waiting to be discovered.

Some of the complexity assumptions are standard (such as $\textup{P}\neq \textup{NP}$), while others may be more controversial (such as the Strong Exponential-Time Hypothesis (SETH)). Therefore, the reader may wonder about the usefulness of proving conditional lower bounds based on unproven assumptions. It is important to point out that these conditional lower bounds are valuable even if we have doubts about the validity of the assumptions. Suppose we use a complexity assumption $X$ to prove that a certain type of algorithm does not exists for a specialized problem $P$, perhaps in an application domain such as database theory. Even if we do not believe in the validity of the assumption $X$, this conditional lower bound shows that the difficulties we face when attacking problem $P$ have nothing to do with the specific details of problem $X$ or the application area: we are really facing assumption $X$ in disguise and we need to disprove that first before any progress can be made on problem $P$. In other words, the conditional lower bound shows that we can stop trying to obtain the desired algorithm for problem $P$, as any such effort would be better spent on trying to disprove the (typically more fundamental) assumption $X$.

The purpose of this article is to highlight some of the lower bound techniques and show what kind of results they can deliver in the context of database theory. We will introduce a number of assumptions, contrast them, and show examples of their use. It has to be emphasized that this article does not aim to be an up to date survey of lower bounds in the area of database theory. The focus is more on the diverse set of assumptions and lower bound techniques that exist, rather than on presenting an exhaustive list of applications for each technique. Some of these example applications come directly from the literature on database query evaluation, but others were stated in essentially equivalent forms in other domains: for constraint satisfaction problems (CSP) or for graph-theoretic problems. Therefore, we begin with introducing the terminology for all these domains and then discuss the results using the most appropriate terminology.

\section{The four domains}
Conjunctive query evaluation is a fundamental problem in database theory. This problem can be equivalently seen as a CSP instance and therefore some of the results in the CSP literature are directly relevant. A large part of the CSP literature uses a formulation using the homomorphism of relational structures, which in some special cases degenerate to traditional graph problems. In this section, we introduce the terminology for all these domains and show how they are connected to each other.

\subsection{Database queries}
A
\emph{join query} $Q$ is an expression of the form
\begin{equation*}
R_1(a_{11},\ldots,a_{1r_1})\bowtie\cdots\bowtie R_m(a_{m1},\ldots,a_{mr_m}),
\label{eqn:joinquery}
\end{equation*}
where the $R_i$ are \emph{relation names} with \emph{attributes}
$a_{i1},\ldots,a_{ir_i}$.
Let $A$ be the set of all attributes
occurring in $Q$ and $n = |A|$.  A \emph{database instance} $\textbf{D}$ for
$Q$ consists of a domain $\dom(\textbf{D})$ and relations $R_i(\textbf{D})\subseteq \dom(\textbf{D})^{r_i}$ of arity $r_i$. It is common to
think of the relation $R_i(\textbf{D})$ as a table whose columns are labeled
by the attributes $a_{i1},\ldots,a_{ir_i}$ and whose rows are the
tuples in the relation. The \emph{answer}, or \emph{set of
  solutions}, of the query $Q$ in $\textbf{D}$ is the $n$-ary relation $Q(\textbf{D})$
with attributes $A$ consisting of all tuples $t$ whose projection on
the attributes of $R_i$ belongs to the relation $R_i(\textbf{D})$, for all
$i$. Given the query $Q$ and the database $D$, the task in the \textsc{Join Query} problem is to compute the set $Q(\textbf{D})$. In the \textsc{Boolean Join Query} problem, we only need to decide if $Q(\textbf{D})$ is empty or not. One can also define the counting version of the problem (i.e, compute $|Q(\textbf{D})|$).

The {\em primal graph} of the query has the set $A$ of attributes as vertex set and two variables are adjacent if there is a relation containing both of them. The {\em hypergraph} of the instance is defined similarly: the vertex set is $A$, and each relation $R_i(a_{i1},\dots, a_{ir_i})$ is represented by a hyperedge $\{a_{i1},\dots, a_{ir_i} \}$. 
\subsection{Constraint satisfaction problems}
Constraint satisfaction is
a general framework that includes many standard algorithmic problems
such as satisfiability, graph coloring, database queries, etc. A
constraint satisfaction problem (CSP) instance consists of a set $V$ of
variables, a domain $D$, and a set $C$ of constraints, where each
constraint is a relation on a subset of the variables. The task is to
assign a value from $D$ to each variable in such a way that every
constraint is satisfied.  For example, 3SAT can be interpreted as a CSP instance
where the domain is $\{0,1\}$ and the constraints in $C$ correspond to
the clauses (thus the arity of each constraint is 3).

Formally, an instance $I$ of a {\em constraint satisfaction problem} is a triple $I=(V ,D, C)$,
where:
\begin{itemize}
\item $V$ is a set of variables,
\item $D$ is a domain of values,
\item $C$ is a set of constraints, $\{c_1,c_2,\dots ,c_q\}$.
Each constraint $c_i\in C$ is a pair $\langle
s_i,R_i\rangle$, where:
\begin{itemize}
\item $s_i$ is a tuple of variables of length $m_i$, called the {\em constraint scope,} and
\item $R_i$ is an $m_i$-ary relation over $D$, called the {\em constraint
  relation.}
\end{itemize}
\end{itemize}

For each constraint $\langle s_i,R_i\rangle$ the tuples of $R_i$
indicate the allowed combinations of simultaneous values for the
variables in $s_i$. The length $m_i$ of the tuple $s_i$ is called the
{\em arity} of the constraint. A {\em solution} to a constraint satisfaction
problem instance is a function $f$ from the set of variables $V$ to
the domain $D$ of values such that for each constraint $\langle
s_i,R_i\rangle$ with $s_i = (v_{i_1},v_{i_2},\dots,v_{i_m})$, the
tuple $( f(v_{i_1}), f(v_{i_2}),\dots,f(v_{i_m}))$ is a member of
$R_i$. Given a CSP instance $I$, we can consider the problem of deciding if a solution exists, the problem of finding all solutions, or the problem of counting the number of solutions.

We say that an instance is {\em
  binary} if each constraint relation is binary, that is, $m_i=2$ for
every constraint\footnote{It is unfortunate that while some communities use
  the term ``binary CSP'' in the sense that each constraint is binary
  (as does this dissertation), others use it in the sense
  that the variables are 0-1, that is, the domain size is 2.}. 
The {\em primal graph} (or {\em Gaifman graph}) of a CSP instance $I=(V ,D, C)$ is a graph $G$
with vertex set $V$, where $x,y\in V$ form an edge if and only if there is a
constraint $\langle
s_i,R_i\rangle\in C$ with $x,y\in s_i$. The \emph{hypergraph} of an instance $I=(V,D,C)$ has $V$ as its vertex
set and for every constraint in $C$ a hyperedge that consists of all
variables occurring in the constraint.

Given a join query $Q$ and a database $\textbf{D}$, we can turn the query problem into a CSP instance $I$ in a straightforward way: the domain of $I$ is $\dom(\textbf{D})$, the set of variables correspond to the attributes $A$ of $Q$, and for each relation $R_i$, there is a corresponding constraint $c_i$ on the variables $a_{i1}$, $\dots$, $a_{ir_i}$. It is clear that the tuples in the answer set of $Q$ in $\textbf{D}$ are in one to one correspondence with the solutions of the CSP instance $I$. This establishes a correspondence between the basic algorithmic problems of the two domains.

It is worth pointing out that even though the two problems are equivalent, a large part of CSP research focuses on problem instances where the domain has small constant size and the number of constraints is large (for example, \textsc{3SAT} is such a problem). This has to be contrasted with the typical setting in database theory research where we assume that there are only a small number of attributes and relations have low arity, but the domain can be large and the number of tuples in a relation can be large.

\subsection{Graph problems}\label{sec:graph-problems}
Given a binary CSP instance $I=(V,D,C)$, we can equivalently formulate it as a graph problem. We construct a graph $G$ the following way: let us introduce $|V|\cdot |D|$ vertices $w_{v,d}$ ($v\in V$, $d\in D$) and for ever constraint $c_i=\langle (u,v),R_i\rangle$, let us make $w_{u,d_1}$ and $w_{v,d_2}$ adjacent if and only if $(d_1,d_2)\in R_i$. Let $W_i=\{w_{i,d}\mid d\in D\}$ and consider the partition $\mathcal{P}=\{W_1,\dots, W_{|V|}\}$. We say that a subgraph $H$ of $G$ {\em respects the partition $\mathcal{P}$} if every class of the partition $\mathcal{P}$ contains exactly one vertex of $H$.

Let $f:V\to D$ be a solution of $I$. If we consider the vertices $\{w_{v,f(v)}\mid v\in V\}$, then it is easy to see that they induce a subgraph $H$ that respects $\mathcal{P}$ and isomorphic to the primal graph of $I$. Conversely, it is not difficult to see that if $G$ has a subgraph that respects $\mathcal{P}$ and is isomorphic to the primal graph of $I$, the it describes a solution of $I$. Therefore, the CSP instance can be described by an instance of {\em partitioned subgraph isomorphism}: given graphs $H$ and $G$, and partition $\mathcal{P}$ of $V(G)$ into $|V(H)|$ classes, find a subgraph of $V$ that respects $\mathcal{P}$ and is isomorphic to $H$. This problem is a natural variant of the standard {\em subgraph isomorphism} problem (find a subgraph of $G$ isomorphic to $H$) and, as we have seen, its complexity is tightly connected to the complexity of CSP instances where the primal graph is $H$.

There is another way in which graph-theoretic notions can describe the solutions of a CSP instance. Consider a binary  CSP instance $I=(V,D,C)$ where every constraint $c_i=\langle (u,v),R_i\rangle$ contains the same binary relation $R_i=R$, which we further assume to be symmetric (that is, $(d_1,d_2)\in R$ if and only $(d_2,d_1)\in R$). Let $H$ be the primal graph of $I$ and let $G$ be a graph with vertex set $D$ where $d_1,d_2\in D$ are adjacent if and only if $(d_1,d_2)\in R$. A {\em homomorphism from $H$ to $G$} is a mapping $f:V(H)\to V(G)$ such that if $u$ and $v$ are adjacent in $H$, then $f(u)$ and $f(v)$ are adjacent in $G$. Note that $f$ does not have to be injective (i.e., $f(u_1)=f(u_2)$ is possible) and if $u$ and $v$ are not adjacent, then we {\em do not} require  that $f(u)$ and $f(v)$ be nonadjacent as well. It is easy to see that every solution $f:V\to D$ of $I$ describes a homomophism from $H$ to $G$, in fact, these homomorphisms are in one to one correspondence with the solutions of $I$. Therefore, the complexity of finding a homomorphism from a fixed graph $H$ to the input graph $G$ is tightly connected to the complexity of CSP for instances with primal graph $H$ where the same symmetric relation $R$ appears in every constraint. If the relation $R$ is not symmetric, then a similar connection can be made to the homomorphism problem in directed graphs.

\subsection{Relational structures}
The connection between CSP and graph homomorphisms that we have seen in the previous section has two major limitations: it worked only for binary CSP instances and only if every constraint contained the same relation $R$. These limitations can be removed if we move from graphs to the much more general setting of relational structures.

A {\em vocabulary} $\tau$ is a
finite set of relation symbols of specified arities. The {\em arity} of
$\tau$ is the maximum of the arities of all relational symbols it
contains. A $\tau$-structure $\mathbf{A}$ consists of a finite set $A$
called the {\em universe} of $\mathbf{A}$ and for each relation symbol $R\in
\tau$, say, of arity $k$, a $k$-ary relation $R^\mathbf{A}\subseteq
A^k$.  A {\em
  homomorphism} from a $\tau$-structure $\mathbf{A}$ to a
$\tau$-structure $\mathbf{B}$ is a mapping $h : A \to B$ from the
universe of $\mathbf{A}$ to the universe of $\mathbf{B}$ that
preserves all relations, that is, for all $R\in \tau$, say, of arity
$k$, and all tuples $(a_1,\dots,a_k)\in R^\mathbf{A}$ it holds that
$(h(a_1),\dots,h(a_k))\in R^{\mathbf{B}}$. Note that if $\tau$ contains only a single relational symbol, which has arity $2$, then $\tau$-structures are essentially directed graphs and the homomorphism problem between $\tau$-structures is equivalent to the homomorphism problem on directed graphs.

More generally, we can express every CSP instance $I=(V,D,C)$ as a homomorphism problem the following way. Let $c=|C|$ be the number of constraints. Let the vocabulary $\tau$ contain $c$ symbols $Q_1$, $\dots$  $Q_c$, where symbol $Q_i$ has the same arity $m_i$ as the constraint $c_i=\langle
s_i,R_i\rangle$. We define $\tau$-structure $\mathbf{A}$ over the universe $V$ such that $Q^{\mathbf{A}}_i$ contains only the tuple $s_i$. We define $\tau$-structure $\mathbf{B}$ over the universe $D$ such that $Q^{\mathbf{B}}_i$ is precisely the relation $R_i$ appearing in constraint $c_i$. Now it can be verified that a mapping $f:V\to D$ is a solution of $I$ if and only if $f$ is a homomorphism from $\mathbf{A}$ to $\mathbf{B}$.

\section{Unconditional lower bounds}\label{sec:uncond-lower-bounds}
Ideally, we would like to prove negative results and lower bounds that are incontestably true mathematical statements. However, our inability to prove the  $\textup{P}\neq \textup{NP}$ hypothesis is a major barrier that prevents us from proving most negative statements of interest. For all we know, it is still possible that  $\textup{P}=\textup{NP}$ and we can solve all database query and CSP instances in polynomial time, and hence at the moment we cannot expect to unconditionally prove any result that rules out such algorithms. As long as we are in the classical setting of computation typically studied in computational complexity (algorithm is given an input, needs to compute a yes-no output), there is little hope in proving strong unconditional lower bounds.

We can hope to obtain unconditional lower bounds only if we deviate from the classical setting: for example, the problem involves the cost of accessing the input or the cost of communication. We show a particular, very simple setting in which we have tight unconditional lower bounds. If the task is to compute the answer to a join query, then the size of the answer is obviously an unconditional lower bound on the number of steps needed for computing the answer. This raises the question: what bounds can we give on the size of the answer and are there query evaluation algorithms that match this lower bound?

Formally, let $Q=
R_1(a_{11},\ldots,a_{1r_1})\bowtie\cdots\bowtie R_m(a_{m1},\ldots,a_{mr_m})$ be a join query and 
let $\textbf{D}$ be a database instance for $Q$ such that every relation $R_i(\textbf{D})$ contains at most $N$ tuples. What can we say about the size of the answer? It is easy to see that $N^m$ is an obvious upper bound: every tuple appearing in the answer chooses one of the at most $N$ possibilities in each of the $m$ relations. But this bound is often very far from being tight. For example, for the query $Q=R_1(a_1,a_2)\bowtie R_2(a_1,a_3)\bowtie R_3(a_2,a_3)$, it is known that the upper bound is $N^{3/2}$ instead of $N^{3}$. The fractional number $3/2$ in the exponent of $N$ suggests that obtaining the bound cannot be completely obvious. Still, precise bounds can be obtained in a clean way using known combinatorial techniques.

Let us define the hypergraph $H$ of the query $Q$ the following way: the vertices are the attributes and we introduce a hyperedge $\{a_{i1},\dots,a_{ir_i}\}$ for each relation $R_i(a_{i1},\dots,a_{ir_i})$. A {\em fractional cover} of a hypergraph $H$ is a mapping $f:V(H)\to [0,1]$ such that for every vertex $v\in V(H)$, we have $\sum_{e\in E(H), v\in e}f(e)\ge 1$. That is, $f$ is a weight assignment on the edges such that the total weight put on each vertex $v$ is at least 1. The {\em weight} of $f$ is $\sum_{e\in E(H)}f(e)$ and the \emph{fractional edge cover number $\rho^*(H)$} of $H$ is the minimum weight of a fractional edge cover of $H$. For example, for the query
$Q=R_1(a_1,a_2)\bowtie R_2(a_1,a_3)\bowtie R_3(a_2,a_3)$, the hypergraph $H$ is a triangle and $\rho^*(H)=3/2$ (assigning weight $1/2$ to each edge is a fractional edge cover and a quick analysis shows that this is optimal).

Using a simple application of Shearer's Lemma \cite{MR859293}, which is a purely combinatorial statement on entropy, one can show that $N^{\rho^*(H)}$ is an upper bound on the number of solutions.
\begin{theorem}[Atserias, Grohe, Marx \cite{DBLP:journals/siamcomp/AtseriasGM13}]\label{th:sizeupper}
Let $Q$ be a join query with hypergraph $H$. Let $\mathbf{D}$ be a database for $Q$ where every relation has at most $N$ tuples. Then the answer of $Q$ in $\mathbf{D}$ has size at most $N^{\rho^*(H)}$.
\end{theorem}
Conversely, we can show that $N^{\rho^*(H)}$ is essentially a tight lower bound. As it is usual with lower bound statements, we have to be a bit more careful with the formulation.
\begin{theorem}[Atserias, Grohe, Marx \cite{DBLP:journals/siamcomp/AtseriasGM13}]\label{th:sizelower}
Let $Q$ be a join query with hypergraph $H$. For infinitely many $N\ge 1$, there is a database $\mathbf{D}_N$ for $Q$ where every relation has at most $N$ tuples and the answer of $Q$ in $\mathbf{D}_N$ has size at least $N^{\rho^*(H)}$.
\end{theorem}
Theorem~\ref{th:sizelower} provides an unconditional lower bound for any algorithm computing the full answer of $Q$ (but of course it does not provide any bound on algorithms that just decide whether the answer is empty or compute the size of the answer). Are there algorithms that match this lower bound? The combinatorial proof of Theorem~\ref{th:sizeupper} can be turned into an algorithm with a constant overhead in the exponent, that is, to obtain $N^{\rho^*(H)+O(1)}$ running time. With additional techniques, it is possible to give tight algorithms that tightly  match the lower bound.
\begin{theorem}[\cite{DBLP:conf/icdt/Veldhuizen14,DBLP:journals/jacm/NgoPRR18}]
Let $Q$ be a join query with hypergraph $H$. Let $\mathbf{D}$ be a database for $Q$ where every relation has at most $N$ tuples. Then the answer of $Q$ in $\mathbf{D}$ can be computed in time $O(N^{\rho^*(H)})$.
\end{theorem}  

\section{NP-hardness}\label{sec:np-hardness}

Since its conception and development in the early 70s \cite{DBLP:conf/stoc/Cook71,DBLP:conf/coco/Karp72,DBLP:books/fm/GareyJ79}, NP-hardness has been the main workhorse of providing intractability results for computational problems. The class NP contains decision problems that can be solved in nondeterministic polynomial time. This robust definition covers (the decision version of) most combinatorial and optimization problems of interest. If a problem $P$ is NP-hard, then this means in particular that a polynomial-time algorithm for $P$ would give polynomial-time algorithms for every problem in NP, which we take as strong evidence that such an algorithm is unlikely. We prove NP-hardness of $P$ by giving a \emph{polynomial-time reduction} from a known NP-hard problem $Q$; this reduction shows that a polynomial-time algorithm for $P$ would give a polynomial-time algorithm for $Q$ and hence for every problem in NP.

Most of the problems studied in database theory or in CSP research are obviously NP-hard, as they contain basic hard problems as special cases. Therefore, it is not completely obvious how to ask reasonable questions about NP-hardness where the answer is not trivial. One direction is to consider restricted parameter values. For example, for CSP problems, one may ask if the problem remains NP-hard if we restrict the domain size $|D|$ to 2 (yes, as \textsc{3SAT} is still a special case), or restrict the constraints to binary (yes, \textsc{3-Coloring} is a special case), or we apply both restrictions (no, with $|D|=2$ and binary constraints the problem becomes the polynomial-time solvable \textsc{2SAT}). Restricting the number $|V|$ of variables to any constant, say 10,  gives a polynomial-time solvable special case even if the domain size $|D|$ is arbitrarily large (as we can try the at most $|D|^{10}$ possible assignments in polynomial time).

More generally, we can introduce restrictions on the problem in a systematic way and determine which of the restrictions lead to polynomial-time solvable and NP-hard special cases. In CSP research, a very well studied family of special cases arise from restricting the type of relations that are allowed in the constraints. Formally, let $D$ be a finite domain and let $\mathcal{R}$ be a finite set of relations over $D$. Then we denote by \textsc{CSP($\mathcal{R}$)} the special case of the general CSP problem where the instance is allowed to contain only constraints $c=\langle s,R\rangle$ where $R\in\mathcal{R}$. Equivalently, we can state this restriction in the language of the homomorphism problem in a very compact way. Let $\tau$ be a vocabulary and let $\mathbf{B}$ be a $\tau$-structure. Then $\textsc{HOM}(\_,\mathbf{B})$ is the special case of the general homomorphism problem where the input is a pair $(\mathbf{A},\mathbf{B})$ of $\tau$-structures, where $\mathbf{A}$ is arbitrary.

A classic result of Schaefer \cite{DBLP:conf/stoc/Schaefer78} characterized the complexity of $\textsc{CSP}(\mathcal{R})$ for any fixed finite set $\mathcal{R}$ of relation over the Boolean domain (i.e., $|D|=2$). More precisely, Schaefer's Dichotomy Theorem showed that every such $\textsc{CSP}(\mathcal{R})$ problem is either polynomial-time solvable or NP-hard, and gave a clean characterization of the two cases. For several years, it was an outstanding open problem to prove an analog of this result for larger fixed domains (the Feder-Vardi Conjecture \cite{MR2000e:68063}). After partial progress, the conjecture was resolved in 2016 independently by Bulatov \cite{DBLP:conf/focs/Bulatov17} and Zhuk \cite{DBLP:conf/focs/Zhuk17}.

While these classification results are cornerstones of modern CSP research, translating them into the language of database theory does not give much useful insight. Indeed, these characterization results would concern special cases where the domain of the attributes have constant size and every database relations is of constant size. From the viewpoint of database theory, a more relevant family of special cases can be obtained by restricting the structure of the query. Let us focus on \textsc{Boolean Join Query}, the problem of deciding if the answer set is empty or not.
If we assume, for example, that the primal graph of the query is a tree (acyclic graph), then it is easy to solve the problem in polynomial time, while the problem may remain NP-hard under other restrictions (for example, under the assumption that the primal graph has maximum degree 3 or is a planar graph etc.). Formally, in the language of CSPs, if $\mathcal{G}$ is any class of graphs, we may want to understand the complexity of the problem $\textup{CSP}(\mathcal{G})$, which is CSP under the restriction that the primal graph belongs to class $\mathcal{G}$.

Can we identify every class $\mathcal{G}$ that leads to polynomial-time solvable special cases and use NP-hardness to give evidence of hardness for every other case? It is known that if $\mathcal{G}$ contains only graphs of bounded treewidth, then $\textup{CSP}(\mathcal{G})$ becomes polynomial-time solvable. Treewidth is a combinatorial measure of graphs that can be though of as a number expressing how treelike the graph is: treewidth is 1 if and only if the graph is an acyclic forest, while other fixed values mean that the graph is similar to a tree with each node being
replaced by a small graph. While the formal definition of treewidth
is technical, it models very faithfully the requirements that make the
algorithmic paradigm ``split on small separators and recurse'' work
and its mathematical naturality is further evidenced by the fact that
it was independently discovered in equivalent formulations at least
three times
\cite{halin72,DBLP:journals/jct/BerteleB73,DBLP:journals/jct/RobertsonS84}.
The precise definition of treewidth is not essential for this paper; we include the definition here only for completeness.

\begin{definition}
A \emph{tree decomposition} of a graph $G$ is a pair $(\mathcal{B},T)$ where $T$
is a tree and $\mathcal{B}=\{B_{t} \mid t\in V(T)\}$ is a collection of subsets
of $V(G)$ such that:
\begin{itemize}
\item $\bigcup_{t \in V(T)} B_{t} = V(G)$, 
\item for
each edge $xy \in E(G)$, $\{x,y\}\subseteq B_t$ for some
$t\in V(T)$; 
\item for each $x\in V(G)$ the set $\{ t \mid x \in B_{t} \}$ induces
  a connected subtree of $T$.
\end{itemize}
The \emph{width} of the tree
  decomposition is $\max_{t \in V(T)} \{|B_{t}| - 1\}$.
  The \emph{treewidth} of a graph $G$ is the minimum width over all
  tree decompositions of $G$.  We denote by $\textup{tw}(G)$ the treewidth of
  graph $G$.
\end{definition}
Freuder \cite{Freuder90AA} showed that if the primal graph has bounded treewidth, then the instance can be solved in polynomial time. By now, the result can be obtained by standard dynamic programming techniques on tree decompositions.
\begin{theorem}[Freuder \cite{Freuder90AA}]\label{th:freuder}
For every fixed $k$, a CSP instance $I=(V,D,C)$ can be solved in time $O(|V|\cdot |D|^{k+1})$ if the primal graph has treewidth at most $k$.
\end{theorem}
It follows from Theorem~\ref{th:freuder} that $\textup{CSP}(\mathcal{G})$ is polynomial-time solvable if $\mathcal{G}$ has bounded treewidth, and it is easy to find graph classes $\mathcal{G}$ with unbounded treewidth (e.g., cliques) where the problem remains NP-hard. But, surprisingly, there seem to be cases that are neither polynomial-time solvable or NP-hard, thus a full classification into these two categories does not seem to be possible.

In the introduction, we mentioned that typically we expect that the problem at hand can be eventually classified as either polynomial-time solvable or NP-hard. While this may be true in most cases, there is no mathematical reason why this should be true in general. In fact, Ladner's Theorem \cite{DBLP:journals/jacm/Ladner75} states that if $\textup{P}\neq \textup{NP}$, then there are NP-intermediate problems in the class NP: problems that are neither polynomial-time solvable nor NP-hard. The proof of Ladner's Theorem produces NP-intermediate problems that are highly artificial, so it is a different question whether there are natural problems that are NP-intermediate. There are two problems that are often highlighted as natural candidates for being NP-intermediate: \textsc{Graph Isomorphism} and \textsc{Integer Factorization}. These two problems are not expected to be polynomial-time solvable, and the fact that they can be solved much more efficiently than brute force \cite{DBLP:conf/stoc/Babai16,DBLP:conf/stoc/BabaiL83,MR1137100} suggests that they are not NP-hard either.

One could say that the reason why \textsc{Graph Isomorphism} and \textsc{Integer Factorization} are NP-intermediate is that the deep algebraic and number-theoretic structures underlying these problems make them occupy a special place in the complexity landscape of NP problems. However, it is important to point out that problems can be (probably) NP-intermediate for more pedestrian reasons: it is possible to scale down an NP-hard problem in a way that it no longer NP-hard, but still not sufficiently easy to be polynomial-time solvable. We will refer to the following (artificial) example also in later sections.
\begin{definition}\label{def:special}
A graph $G$ is {\em special} if it has exactly two connected components: a clique of size $k$ for some integer $k\ge 1$ and a path of exactly $2^k$ vertices. \textsc{Special CSP} and \textsc{Special Boolean Join Query} are the restricted cases of the general problems where we assume that the primal graph is special.
\end{definition}
Let us give some intuitive arguments why these problems could be NP-intermediate (we will make this more formal in later sections). First, the path part can be solved efficiently in polynomial time. Then we need to solve the clique part, which can certainly be done by brute force in time $O(n^k)$, where $n$ is the total length of the input. But as already the primal graph has size larger than $2^k$, we have $n\ge 2^k$ and hence $k\le \log n$. Thus we can solve the problem in quasipolynomial time $n^{O(\log n)}$, which would be an exceptionally unusual property of an NP-hard problem. Moreover, it is not clear what substantial improvements we can expect on this algorithm: one would need to solve the clique part significantly faster than brute force. Therefore, it seems that these problems variants are likely to be NP-intermediate with best possible running time around $n^{O(\log n)}$.

This example shows that even if we just want to understand which special cases are polynomial-time solvable, then NP-hardness may not be sufficient for this purpose. As we shall see in later sections, we need to use other lower bound techniques for this type of classification. Additionally, these lower bound techniques can provide stronger lower bounds beyond just ruling out polynomial-time algorithms, showing the optimality of certain algorithms in a tighter way.

\section{Parameterized Intractability}
\label{sec:param-intr}

Parameterized complexity considers algorithmic problems where each
input instance has a parameter $k$ associated with it. This parameter
is typically either the size of the solution we are looking for or
some measure of the input, such as the number of variables in a formula, the dimension of the input point
set, the maximum degree of the input graph, or perhaps the alphabet size of the
input strings. The central goal of parameterized complexity is to
develop algorithms that are efficient on instances where the value of
the parameter is small. Formally, we say that a parameterized problem
is {\em fixed-parameter tractable} (FPT) if it can be solved in time
$f(k)\cdot n^{O(1)}$, where $n$ is the size of the input and $f$ is a
computable function depending only on $k$. This form of running time
has to be contrasted with the running time $n^{O(k)}$ of brute force
algorithms that are often easily achievable if $k$ is the size of the
solution we are looking for. Research in the past three decades has
shown that many of the natural NP-hard problems are FPT with various
parameterizations, leading to algorithms that are often highly
nontrivial and combinatorially deep \cite{CyganFKLMPPS15,DBLP:series/txcs/DowneyF13,MR2001b:68042,MR2238686}.

As an example, let us consider \textsc{Vertex Cover}: given a graph $G$ and an integer $k$, the task is to find a vertex cover $S$ of size at most $k$, that is, a set $S$ of at most $k$ vertices such that every edge of $G$ has at least one endpoint in $S$. Clearly, we can solve the problem by brute force on an $n$-vertex graph by trying each of the $O(n^k)$ sets $S$ of size at most $k$. But more efficient algorithms are available: a standard application of the bounded-depth search tree technique already delivers a $2^k\cdot n^{O(1)}$ algorithm, which can be further refined with additional techniques \cite{MR2002h:05149}. This means that \textsc{Vertex Cover} is FPT parameterized by the size of the solution. For \textsc{Clique}, the trivial $O(n^k)$ brute force search can be improved to about $O(n^{\omega k/3})$) (where $\omega<2.3729$ is the exponent for matrix multiplication \cite{DBLP:conf/soda/AlmanW21,nevsetvril1985complexity,nevsetvril1985complexity}), but no FPT algorithm is known despite significant efforts.

Motivated by this apparent difference between \textsc{Vertex Cover} and \textsc{Clique}, Downey and Fellows introduced the notion of W[1]-hardness and the $\textup{FPT}\neq \textup{W[1]}$ hypothesis \cite{MR2001b:68042}. We omit here the technical definitions related to the class W[1]; for the purpose of proving negative evidence, it is sufficient to know that $\textup{FPT}\neq \textup{W[1]}$ is equivalent to the statement ``\textsc{Clique} is not FPT'' (or to 
``\textsc{Independent Set} is not FPT'', as the two problems are equivalent by taking the complement of the graph). To define W[1]-hardness, we need first the following notion of reduction:
\begin{definition}
  Let $P$ and $Q$ be two parameterized problems. A {\em parameterized reduction} transforms an instance $x$ of $P$ with parameter $k$ to an instance $x'$ of $Q$ with parameter $k'$ such that
  \begin{enumerate}
  \item $(x,k)$ is a yes-instance of $P$ if and only if $(x',k')$ is a yes-instance of $Q$.
  \item The running time of the reduction is $f(k)|x|^{O(1)}$ for some computable function $f$.
    \item We have $k'\le f(k)$ for some computable function $f$.
  \end{enumerate}
\end{definition}
The third requirement is what makes this notion very different from usual polynomial-time reductions: we have to pay extra attention not to blow up too much the parameter $k$ in the reduction. Parameterized reductions were designed in a way that they transfer the property of being FPT: it can be shown that if there is a parameterized reduction from $P$ to $Q$ and $Q$ is FPT, then $P$ is FPT as well. We can define W[1]-hardness by saying that a problem $P$ is W[1]-hard if there is a parameterized reduction from \textsc{Clique} to $P$. We can interpret this as evidence that $P$ is not FPT: an FPT algorithm for $P$ would show that \textsc{Clique} is FPT,  violating the $\textup{FPT}\neq \textup{W[1]}$ hypothesis.

Let us have a look at the complexity of CSP via the lens of parameterized complexity. Given a instance $I=(V,D,C)$, we can introduce the number $k=|V|$ of variables as the parameter of the instance. We can decide if there is a solution by trying each of the $|D|^k=O(n^k)$ possible assignments. The NP-hardness of the problem implies that this cannot be improved to $n^{O(1)}$ (assuming $\textup{P}\neq\textup{NP}$), but this does not rule out the possibility that the problem is FPT parameterized by $k$, that is, there is a $f(k)\cdot n^{O(1)}$ time algorithm. Such an algorithm would be certainly of interest in contexts where we can assume that $k$ is small, but $D$ is large (which is typically true in database applications). However, it is easy to see that the problem of finding a clique of size $k$ in a graph $G$ can be expressed as a CSP problem with $k$ variables, $\binom{k}{2}$ constraints, and domain $D=V(G)$. That is, there is a parameterized reduction from \textsc{Clique} to \textsc{CSP} parameterized by the number of variables, showing that the latter problem is unlikely to be FPT either.

We can now return to the question left open at the end of Section~\ref{sec:np-hardness}: what are those classes $\mathcal{G}$ for which $\textsc{CSP}(\mathcal{G})$ is polynomial-time solvable? As we have seen, NP-hardness does not seem to be sufficiently strong to highlight all the negative cases. Instead, let us look at the fixed-parameter tractability of $\textsc{CSP}(\mathcal{G})$, parameterized by the number $k$ of variables. Now if $\textsc{CSP}(\mathcal{G})$ can be proved to be W[1]-hard for some $\mathcal{G}$, then this implies in particular that it is not polynomial-time solvable, assuming $\textup{FPT}\neq \textup{W[1]}$.

For example, let $\mathcal{G}$ contain every special graph, as defined in Definition~\ref{def:special}, and let us consider \textsc{Special CSP} (that is, $\textsc{CSP}(\mathcal{G})$) parameterized by the number of variables. Given an instance of \textsc{Clique} (a graph $G$ where it has to be decided if there is a clique of size $k$), then we can express it as a \textsc{Special CSP} instance as follows. We introduce $k$ variables connected by $\binom{k}{2}$ binary constraints to express the problem of finding a $k$-clique, and additionally we introduce $2^k$ dummy variables connected by constraints forming a path. The primal graph is a $k$-clique plus a path on $2^k$ vertices, as required in \textsc{Special CSP}. The reduction turns the problem of finding a $k$-clique to a \textsc{Special CSP} instance on $f(k)=k+2^k$ variables, hence this is a proper parameterized reduction. It follows that \textsc{Special CSP} is W[1]-hard parameterized by the number of variables, and hence unlikely to be polynomial-time solvable.

More generally, Grohe, Schwentick, and Segoufin \cite{380867} proved that if $\mathcal{G}$ has unbounded treewidth, then $\textsc{CSP}(\mathcal{G})$ is W[1]-hard, leading to a complete classification.
\begin{theorem}[Grohe, Schwentick, and Segoufin~\cite{380867}]\label{th:csptw}
  Let $\mathcal{G}$ be a decidable class of graphs. Assuming $\textup{FPT}\neq \textup{W[1]}$, the following are equivalent:
  \begin{enumerate}
  \item $\mathcal{G}$ has bounded treewidth,
  \item     $\textsc{CSP}(\mathcal{G})$ is polynomial-time solvable,
    \item $\textsc{CSP}(\mathcal{G})$ is FPT parameterized by the number $k$ of variables.
  \end{enumerate}
\end{theorem}
Observe that there is a major coincidence here: the polynomial-time solvable cases are exactly the same as the FPT cases (which in principle could have been a more general class). This coincidence makes it possible to use W[1]-hardness to identify those cases that are not polynomial-time solvable.

A more general result gives a classification in the framework of homomorphism problem for relational structures (which can be directly translated to results for \textsc{Boolean Join Query}). Let $\tau$ be a vocabulary and let $\mathcal{A}$ be a class of $\tau$-structures. Then $\textsc{HOM}(\mathcal{A},\_)$ is the special case of the general homomorphism problem where given two $\tau$-structures $(\mathbf{A},\mathbf{B})$ with $\mathbf{A}\in\mathcal{A}$ and $\mathbf{B}$ arbitrary, the task is to decide if there is a homomorphism from $\mathbf{A}$ to $\mathbf{B}$. The polynomial-time solvable cases again depend on treewidth, but in a slightly more complicated manner. If $\mathbf{A}'$ is a substructure of $\mathbf{A}$ such that there is a homomorphism from $\mathbf{A}$ to $\mathbf{A'}$, then the problem instances $(\mathbf{A},\mathbf{B})$ and $(\mathbf{A}',\mathbf{B})$ are equivalent. The smallest such substructure $\mathbf{A}'$ of $\mathbf{A}$ is called the {\em core} of $\mathbf{A}$ (it is known to be unique up to isomorphism). It is the treewidth of this core that determines the complexity of the problem.

\begin{theorem}[Grohe~\cite{1206036}]\label{th:homtw}
  Let $\tau$ be a finite vocabulary and let $\mathcal{A}$ be a decidable class of $\tau$-structures. Assuming $\textup{FPT}\neq \textup{W[1]}$, the following are equivalent:
  \begin{enumerate}
  \item the cores of the structures in $\mathbf{A}$ have bounded treewidth,
  \item  $\textsc{HOM}(\mathcal{A},\_)$ is polynomial-time solvable,
    \item$\textsc{HOM}(\mathcal{A},\_)$ is FPT parameterized by the size $k$ of the universe of $\mathbf{A}$.
  \end{enumerate}
\end{theorem}
Again, we have that the polynomial-time solvable and FPT cases coincide, allowing the use of W[1]-hardness for the classification of both properties.


\section{The Exponential-Time Hypothesis}\label{sec:expon-time-hypoth}
Parameterized complexity gives a finer understanding of the complexity of problems: for example, the negative results not only tell us that \textsc{Clique} is not polynomial-time solvable, but they rule out algorithms with running time $f(k)n^{O(1)}$. However, this is still a qualitative result that rules out a certain running time, but does not tell us the exact complexity of the problem: can the $n^k$ brute force search be improved to, say, $n^{\sqrt{k}}$, or to $n^{O(\log k)}$, or to $n^{O(\log\log\log\log k)}$, or to $\ldots$? For all we know, such algorithms cannot be ruled out based on the $\textup{P}\neq \textup{NP}$ or $\textup{FPT}\neq\textup{W[1]}$ conjectures. Similarly, for FPT problems such as \textsc{Vertex Cover} where the best known algorithms have running time of the form $2^{O(k)}\cdot n^{O(1)}$, we cannot rule out that these algorithms can be significantly improved to, say, $2^{O(\sqrt{k})}\cdot n^{O(1)}$.

The Exponential-Time Hypothesis (ETH), formulated by Impagliazzo, Paturi, and Zane \cite{ImpagliazzoP01,MR1894519}, makes the assumption $\textup{P}\neq \textup{NP}$ more quantitative: informally, it not only tells us that NP-hard problems do not have polynomial-time algorithms, but it postulates that NP-hard problems really require exponential time and cannot be solved in subexponential time. The formal statement of the ETH is somewhat technical and for most applications it is more convenient to use the following assumption instead, which is an easy consequence of the ETH:

\begin{hypothesis}[Consequence of the ETH, Impagliazzo, Paturi, and Zane \cite{ImpagliazzoP01,MR1894519}]\label{ass:eth}
\textsc{3SAT} with $n$ variables cannot be solved in time $2^{o(n)}$.
\end{hypothesis}
\textsc{3SAT} is the fundamental satisfiability problem
where, given a Boolean formula in conjunctive normal form with at most
3 literals in each clause (e.g.,
$(x_1 \vee \bar x_3 \vee x_5) \wedge (\bar x_1 \vee x_2 \vee x_3)
\wedge (\bar x_2 \vee x_3 \vee x_4)$), the task is to decide whether a
satisfying assignment exists.  For completeness, let us recall the
formal statement of the ETH, of which Hypothesis~\ref{ass:eth} is an
easy consequence.  Let $s_k$ be the infinum of all real numbers
$\delta$ for which there exists an $O(2^{\delta n})$ time algorithm
for \textsc{$k$-SAT}. Then the ETH is the assumption that $s_k>0$ for
every $k\ge 3$.  It is easy to show that this assumption implies
Hypothesis~\ref{ass:eth}, hence if we can show that some statement
would refute Hypothesis~\ref{ass:eth}, then it would refute the ETH as
well.

As \textsc{3SAT} can interpreted as a special case of CSP with domain size 2 and constraints of arity 3, we can translate Hypothesis~\ref{ass:eth} into the language of CSPs to obtain a lower bound for solving instances with constant domain size.
\begin{corollary}
Assuming the ETH, there is no algorithm that solves every CSP instance $I=(V,D,C)$ in time $2^{o(|V|)}\cdot n^{O(1)}$, even if $|D|=2$ and every constraint has arity at most 3.
\end{corollary}

Hypothesis~\ref{ass:eth} rules out the existence of algorithms that are subexponential in the number $n$ of variables. But the number $m$ of clauses in a \textsc{3SAT} instance can be up to cubic in the number of variables, thus the length of the instance can be much larger than $O(n)$. Therefore, Hypothesis~\ref{ass:eth} does not rule out the existence of algorithms that are subexponential in the length of the instance: it could be potentially the case that all the really hard instances of \textsc{3SAT} have, say, $\Omega(n^2)$ clauses, hence a $2^{o(\sqrt{n+m})}$ algorithm would be still compatible with Hypothesis~\ref{ass:eth}. Impagliazzo, Paturi and Zane \cite{MR1894519} showed that this is not the case: the Sparsification Lemma implies that, for the purposes of Hypothesis~\ref{ass:eth}, \textsc{3SAT} remains hard already when restricted to instances with a linear number of clauses. With the Sparsification Lemma, the following stronger assumption follows from Hypothesis~\ref{ass:eth}:
\begin{hypothesis}[Consequence of the ETH + Sparsification Lemma, Impagliazzo, Paturi, and Zane \cite{MR1894519}]\label{ass:eth2}
\textsc{3SAT} with $n$ variables and $m$ clauses cannot be solved in time $2^{o(n+m)}$.
\end{hypothesis}
This stronger assumption turns out to be very useful to prove lower
bounds for other problems. Reductions from \textsc{3SAT} to other
problems typically create instances whose size depends not only on the
number $n$ of variables, but also on the number $m$ of clauses, hence
it is important to have lower bounds on \textsc{3SAT} in terms of both
$n$ and $m$. For example, if we look at textbook reductions from \textsc{3SAT} to \textsc{3-Coloring}, then they transform a formula with $n$ variables and $m$ clauses into a graph with $O(n+m)$ vertices and $O(n+m)$ edges. Such a reduction together with Hypothesis~\ref{ass:eth2} implies a lower bound for binary CSP over a constant domain size. 
\begin{corollary}
Assuming the ETH, there is no algorithm that solves every CSP instance $I=(V,D,C)$ in time $2^{o(|V|+|C|)}\cdot n^{O(1)}$, even if $|D|=3$ and every constraint is binary.
\end{corollary}

Let us turn our attention now to parameterized problems. A key result in parameterized complexity states that, assuming ETH, the $n^k$ brute force search for \textsc{Clique} cannot be improved better than a constant factor in the exponent, even if we allow an arbitrary $f(k)$ factor in the running time.
\begin{theorem}[Chen et al.~\cite{MR2121603}]\label{th:chenclique}
  Assuming ETH, \textsc{Clique} cannot be solved in time $f(k)\cdot n^{o(k)}$ for any computable function $f$.
\end{theorem}
The same is true for the \textsc{Partitioned Clique}, which, as we have seen in Section~\ref{sec:graph-problems}, is essentially equivalent to a binary CSP instance where the primal graph is clique. Therefore, we can translate Theorem~\ref{th:chenclique} into the language of CSPs.
\begin{theorem}\label{th:cspclique}
  Assuming ETH, there is no algorithm that solves every binary CSP instance $I=(V,D,C)$ in time $f(|V|)\cdot |D|^{o(|V|)|}\cdot n^{O(1)}$, where $f$ is an arbitrary computable function.
\end{theorem}
Moreover, we have seen in Section~\ref{sec:param-intr} that \textsc{$k$-Clique} can be reduced to a \textsc{Special CSP} instance with $k+2^k$ variables. Together with Theorem~\ref{th:chenclique} it follows that, assuming the ETH, there is no $f(|V|)n^{o(\log |V|)}$ time algorithm for \textsc{Special CSP}. This makes the NP-intermediate status of the problem very precise: it is indeed $n^{O(\log |V|)}$ the best possible running time we can hope for this problem.

   Note that the treewidth of a $k$-clique is $k-1$. Therefore, Theorem~\ref{th:cspclique} shows that if $k$ is the treewidth of the primal graph of the CSP instance, then the $n^{O(k)}$ time algorithm of Freuder~\cite{Freuder90AA} is essentially optimal in the sense that the exponent cannot be improved by more than a constant factor.

   \begin{theorem}\label{th:csptweth}
   Assuming ETH, there is no algorithm that solves every binary CSP instance $I=(V,D,C)$ in time $f(|V|)\cdot  n^{o(k)}$, where $k$ is the treewidth of the primal graph and $f$ is an arbitrary computable function $f$.
   \end{theorem}  

  One could interpret Theorem~\ref{th:csptweth} as saying that the treewidth-based algorithm of Freuder is an optimal way of solving CSP instances. However, this interpretation is misleading. What Theorem~\ref{th:csptweth} really says is that there is one type of primal graphs, namely cliques, where the $n^{O(k)}$ running time that follows using that treewidth-based algorithm is essentially optimal. This {\em does not} rule out the possibility that there are some graph classes, maybe planar graphs, bounded-degree graphs, interval graphs, etc. where it is possible to solve the problem $n^{o(k)}$ time, where $k$ is the treewidth of the primal graph. Formally, we can approach this possibility in the spirit of Theorem~\ref{th:csptw}, by considering the problem $\textsc{CSP}(\mathcal{G})$, where the primal graph is restricted some class $\mathcal{G}$. The following lower bound shows that the treewidth-based algorithm is still optimal for any such $\textsc{CSP}(\mathcal{G})$, up to a logarithmic factor in the exponent.
  \begin{theorem}[\cite{marx-toc-treewidth}]\label{th:beattw}
    Let $\mathcal{G}$ be a class of graphs with unbounded treewidth. Assuming ETH, there is no algorithm that solves every instance $I=(V,D,C)$ of $\textsc{CSP}(\mathcal{G})$ in time $f(|V|)n^{o(k/\log k)}$, where $k$ is the treewidth of the primal graph and $f$ is an arbitrary computable function.
  \end{theorem}
  The formulation of Theorem~\ref{th:beattw} was originally chosen in a way to be analogous to the formulation of Theorem~\ref{th:csptw}. However, it later turned out that it is possible to state it in a slightly more robust and expressive way that shows the precise complexity of individual primal graphs.
\begin{theorem}[\cite{jacm-genus}]\label{T:beattreewidth}
    Assuming the ETH, there exists a universal constant $\alpha>0$ such that for any fixed primal graph $G$ with treewidth $k\geq 2$, there is no algorithm deciding the binary CSP instances $I=(V,D,C)$ whose primal graph is $G$ in time $O(|D|^{\alpha \cdot k/\log k})$.
\end{theorem}

\section{The Strong Exponential-Time Hypothesis}\label{sec:strong-expon-time}

Despite the usefulness of the ETH, there are complexity lower bounds
that seem to be beyond the reach of what can be proved as a
consequence of this hypothesis. Impagliazzo, Paturi, and Zane \cite{MR1894519} proposed an even stronger assumption on the complexity of NP-hard problems: the so-called Strong Exponential-Time Hypothesis (SETH). Using the notation introduced at the beginning of Section~\ref{sec:expon-time-hypoth}, the SETH assumes that $\lim_{k\to\infty}s_k=1$. The following consequence of the SETH is a convenient formulation that can be used as a starting point for lower bounds on other problems:
\begin{hypothesis}[Consequence of the SETH, Impagliazzo, Paturi, and Zane \cite{MR1894519}]\label{ass:seth}
\textsc{SAT} with $n$ variables and $m$ clauses cannot be solved in time $(2-\epsilon)^n \cdot m^{O(1)}$ for any $\epsilon>0$.
\end{hypothesis}
Intuitively, Hypothesis~\ref{ass:seth} states that there is no better algorithm for \textsc{SAT} than the brute force search of trying each of the $2^{n}$ possible assignments. 
Note that here \textsc{SAT} is the satisfiability problem with unbounded clause length. For fixed clause length, algorithms better than $2^{n}$ are known: for example, the best known algorithms for \textsc{3SAT} and \textsc{4SAT} have running times $1.308^n$ and $1.469^n$, respectively \cite{DBLP:journals/siamcomp/Hertli14}. The SETH states that the base of the exponent has to get closer and closer to~1 as the clause length increases, and it is not possible to have an algorithm with base $2-\epsilon$ that works for arbitrary large clause length.

It is important to note that there is no known analogue of the Sparsification Lemma for the SETH. That is, we cannot assume that the hard instances stipulated by Hypothesis~\ref{ass:seth} have only a linear number of clauses: for all we know, the number of clauses can be exponential in the number $n$ of variables. This severely limits the applicability of lower bounds based on the SETH as any reduction from the \textsc{SAT} instance would create instances whose sizes are potentially exponentially large in $n$. Nevertheless, the SETH has found applications in parameterized complexity, for example, giving tight lower bounds on how the running time has to depend on treewidth \cite{DBLP:journals/talg/LokshtanovMS18,BorradaileL15,CurticapeanM16,CyganDLMNOPSW16,CyganNPPRW11,DBLP:conf/esa/IwataY15,DBLP:conf/stacs/EgriMR18,JaffkeJ17}.

An important parameterized problem for which the SETH gives a very tight lower bound is \textsc{Dominating Set}. The {\em closed neighborhood} $N[v]=N(v)\cup \{v\}$ of a vertex $v$ consists of the vertex itself and its neighbors. A {\em dominating set} $S$ is a set of vertices that contains a vertex from the closed neighborhood of every vertex, in other words, every vertex is either selected or has a selected neighbor. In the \textsc{Dominating Set} problem,  given a graph $G$ and an integer $k$, the task is to find a dominating set $S$ of size at most $k$. If $G$ is an $n$-vertex graph, then the trivial brute force algorithm enumerates  the $O(n^k)$ subsets of size at most $k$ and needs $O(n^2)$ time for each of them to check if they form a solution. This results in a $O(n^{k+2})$ time algorithm, which can be improved to $n^{k+o(1)}$ \cite{DBLP:journals/tcs/EisenbrandG04}. The following lower bound shows that any small constant improvement beyond $k$ in the exponent would violate the SETH.
\begin{theorem}[Patrascu and Williams \cite{patrascu10sat-lbs}]\label{th:domseth}
If there is an integer $k\ge 3$ and a real number $\epsilon >0$ such that \textsc{$k$-Dominating Set} can be solved in time $O(n^{k-\epsilon})$ on $n$-vertex graphs, then the SETH is false.  
\end{theorem}
As a demonstration, we show how 
Theorem~\ref{th:domseth} allows us to make Theorem~\ref{th:csptweth} tighter. We are not just ruling out $|D|^{o(k)}$ time, but any potential improvement in the exponent of the domain size beyond $k$, getting closer to the upper bound of Theorem~\ref{th:freuder}.
\begin{theorem}
  If there are integers $k\ge 3$, $c\ge 1$, and a real number $\epsilon >0$ such that there is an algorithm solving CSP instances $I=(V,D,C)$ whose primal graph has treewidth at most $k$ in time $O(|V|^c\cdot |D|^{k-\epsilon})$, then the SETH is false.
\end{theorem}  
\begin{proof}
  Let  $g\ge 1$ be the smallest integer such that $g\epsilon>c+\epsilon$ and let $t=gk$.
First we present a generic reduction from \textsc{$t$-Dominating Set} on an $n$-vertex graph to a CSP instance with domain size $n$ where the treewidth of the primal graph is $t$. Let $G$ be an $n$-vertex graph where we need to find a solution $S$ of size at most $t$. For simplicity of notation, let us assume that $V(G)=[n]$. We construct a CSP instance $I=(V,D,C)$ the following way. The set $V$ contains $t+n$  variables $s_1$, $\dots$, $s_t$, $x_1$, $\dots$, $x_n$. The domain $D$ is $V(G)=[n]$. The intended meaning of the value of $s_i$ is the $i$-th vertex of the solution, and the intended meaning of $x_j=i$ is that the solution vertex represented by $s_i$ is in $N[j]$. To enforce this interpretation, for every $i\in [t]$ and $j\in[n]$, we introduce a constraint $c_{i,j}=\langle (s_i,v_j),R_{i,j}\rangle$, where
  \begin{align*}
    R_{i,j}=&\left\{ (a,b) \mid a \in [n], b\in [t], b\neq i \right\} \bigcup \\
&    \left\{ (a,b) \mid a \in [n], b\in [t], b= i, a\in N[j] \right\}
  \end{align*}
  It is not difficult to check that if there is a solution to this CSP instance, then $\{s_1,\dots,s_t\}$ is a dominating set (as vertex $s_{x_j}$ is in the closed neighborhood $N[j]$ of $j$). Conversely, a solution $S$ of \textsc{$t$-Dominating Set} can be turned into a solution of this CSP instance. Observe that the  primal graph is complete bipartite graph with $t$ vertices on one side and $n$ vertices on the other side. Such a graph has treewidth at most $t$.

  To obtain the required form of the lower bound, we need to modify the constructed CSP instance. Let us group the variables $x_1$, $\dots$, $x_t$ into $t/g=k$ groups of size $g$ each. If we increase the domain from $D$ to $D^g$ (having size $n^g$), we can represent each group with a single new variable (and modify the constraints accordingly). This way, we can obtain an equivalent CSP instance $I'=(V',D',C')$ where $|D'|=n^g$ and treewidth of the primal graph is at most $k$. By our assumption, this CSP instance $I'$ (and hence the original \textsc{$t$-Dominating Set} instance) can be solved in time

  \begin{align*}
    O(|V'|^c\cdot |D'|^{k-\epsilon})&=
    O(n^c\cdot n^{g( t/g-\epsilon)})\\
    &=O(n^{t+c-g\epsilon})=O(n^{t-\epsilon}).
  \end{align*}
  The size of the constructed instance $I'$ can be generously bounded by $O(n^{2g+1})$ which is less than $O(n^{t-\epsilon})$. The reduction presented above can be done in time linear in the size of $I'$, hence the running time of the reduction itself is dominated by the running time of solving $I'$ with the assumed algorithm. Therefore, we obtain an algorithm for solving \textsc{$t$-Dominating Set} in time $O(n^{t-\epsilon})$.
  By Theorem~\ref{th:domseth}, this violates SETH.
\end{proof}

In recent years, the SETH has been successfully used to give lower bounds for polynomial-time solvable problems, for example, by showing that the textbook $O(n^2)$ dynamic programming algorithm for \textsc{Edit Distance} cannot be significantly improved: it cannot be solved in time $O(n^{2-\epsilon})$ for any $\epsilon>0$, unless the SETH fails \cite{DBLP:conf/focs/BringmannK15,DBLP:journals/siamcomp/BackursI18}. Many other tight results of this form can be found in the recent literature under the name ``fine-grained complexity'' \cite{DBLP:conf/soda/BringmannGMW18,DBLP:conf/soda/BringmannK18,DBLP:conf/stoc/AbboudBDN18,DBLP:conf/focs/BringmannK15,DBLP:conf/focs/Bringmann14,DBLP:journals/siamcomp/BackursI18,AbboudBW15,DBLP:journals/siamcomp/AbboudWY18,Williams15,patrascu10sat-lbs,RodittyW13}.

\section{Other conjectures}
We finish the overview of lower bound techniques with a few other complexity conjectures that have appeared recently in the literature. This section does not contain any strong results or nontrivial reductions; the goal is to present assumptions that have direct consequences when translated into the language of database theory and CSP. The aim is to raise awareness of the existence of these conjectures, which may be the starting point of future research.

\textbf{The $k$-clique conjecture.} Matrix multiplication techniques can be used to detect if a graph contains a triangle: if $A$ is the adjacency matrix of $G$, then $G$ contains a triangle if and only if $A^3$ has a nonzero value on the diagonal. Therefore, if we have an algorithm for multiplying two $n\times n$ matrices in time $O(n^\omega)$ for some $\omega$ (the current best known algorithm has $\omega<2.3729$ \cite{DBLP:conf/soda/AlmanW21}), then we can detect in time $O(n^3)$ if an $n$-vertex graph contains a triangle. Nesetril and Poljak \cite{nevsetvril1985complexity} showed in 1985 that this can be generalized further for detecting a clique of size $k$ in time $O(n^{\omega k/3})$ (if $k$ is divisible by $3$). As no significant improvement over this approach appeared in the past 35 years, one can conjecture that there is no $O(n^{(\omega-\epsilon)k/3+c})$ time algorithm for \textsc{$k$-Clique} for any $\epsilon,c>0$. Abboud, Backurs, and Vassilevska Williams \cite{DBLP:journals/siamcomp/AbboudBW18} used this conjecture to give evidence that Valiant's $O(n^\omega)$ time parsing algorithm \cite{DBLP:journals/jcss/Valiant75} from 1975  is optimal.

As discussed in Section~\ref{sec:graph-problems}, \textsc{$k$-Clique} on an $n$-vertex graph can be represented as a CSP with $k$ variables, $\binom{k}{2}$ constraints, and domain size $n$. Therefore, the $k$-clique conjecture further refines
Theorem~\ref{th:cspclique} by ruling out not only a $|D|^{o(|V|)}$ dependence in the running time, but also $|D|^{(\omega-\epsilon)|V|/3+c}$ for any $\epsilon,c>0$.

\textbf{The $d$-uniform hyperclique conjecture.} A hypergraph is $d$-uniform if every hyperedge contains exactly $d$ vertices. The analog of a $k$-clique in a $d$-uniform hypergraph is a set $S$ of $k$ vertices such that each of the $\binom{k}{d}$ possible hyperedges are present in $S$. Somewhat surprisingly, matrix multiplication techniques seem to speed up the search for $k$-cliques only for $d=2$ (ordinary graphs). For any fixed $d\ge 3$, nothing
substantially better is known than trying every set of size $k$. This suggests the conjecture that there is no $O(n^{(1-\epsilon)k+c})$ time algorithm for detecting $k$-cliques in $d$-uniform hypergraphs for any fixed $d\ge 2$ and $\epsilon,c>0$ \cite{DBLP:conf/soda/LincolnWW18}. We can again translate this conjecture into the language of CSPs: we can show that even if the arity of every constraint is at most 3, there is no $f(|V|)\cdot |D|^{(1-\epsilon)|V|+c}\cdot n^{O(1)}$ algorithm for CSP for any $\epsilon,c>0$ and computable function $f$. Therefore, we get very tight lower bounds showing that essentially the brute force search of all assignments cannot be avoided. The $d$-uniform hyperclique conjecture was used to rule out the possibility of constant-delay enumeration algorithms \cite{DBLP:conf/csl/BaganDG07,DBLP:phd/hal/BraultBaron13}.

\textbf{The triangle conjecture.} In database query problems it is more relevant to express the running time in terms of the size of the database relations rather than the size of the domain. For example, given the query $Q=R_1(a_1,a_2)\bowtie R_2(a_1,a_3)\bowtie R_3(a_2,a_3)$ whose primal graph is the triangle, matrix multiplication can be used to check in time $O(d^\omega)$ if the answer is empty, where $d$ is the size of the domain of the attributes. But what can we say about the running time expressed as a function $N$ of the maximum size of the relations? As we have seen in Section~\ref{sec:uncond-lower-bounds}, the size of the solution is $O(N^{3/2})$, we can enumerate it in time $O(N^{3/2})$, and this is tight. However, this does not rule out the possibility that there are faster algorithms for deciding if the answer is empty. In particular, can this be solved in linear time? Note that this question is equivalent to asking for the best possible running time for detecting triangles, where the running time is now expressed as a function of the number $m$ of edges. The best known algorithm of this form detects the existence of a triangle in time $O(m^{2\omega/(w+1)})$ \cite{DBLP:journals/algorithmica/AlonYZ97} and one can state as a conjecture (Strong Triangle Conjecture \cite{DBLP:conf/focs/AbboudW14}) that this is indeed best possible.

\section{Conclusions}
We have seen a sequence of complexity lower bounds of various
strengths.  The results were based on assumptions with different
levels of plausibility, going all the way from unconditional bounds,
classic NP-completeness, to novel conjectures. The results also differ in the tightness of the lower bound: they can be only qualitative results (polynomial-time vs.~NP-hard, FPT vs.~W[1]-hard) or quantitative lower bounds showing the optimality of current algorithms to various levels of tightness.

What can we learn from all these results? First, when aiming for a lower bound, we need to select a precise form of the bound. Ideally, we would like to have negative results that rule out the possibility of any improved algorithm compared to what is known currently, showing that they are already optimal. The meaning of ``any improved algorithm'' needs to be clarified precisely and the choice of this meaning can greatly influence the technical difficulty of the lower bound proof and the required assumptions. Second, we need to choose a suitable complexity assumption that we can base the result on. There are established conjectures, such as the ETH, that are widely used in different domains. But we should be ready to connect our database theory or CSP problem at hand with other, less celebrated open question as well. The general theme of conditional lower bounds is to transform a relatively specialized question to a more fundamental question that was studied from multiple directions. If we can formally establish that the main challenge in understanding our problem is really some other, more fundamental problem in disguise, then this means that spending further efforts on finding improved algorithms is not timely and we can assume for the time being that any algorithm matching the lower bound is optimal. This is a common situation in complexity theory: as it is often said, computational complexity progresses by reducing the number of questions, without increasing the number of answers.

\begin{acks}
 Research supported by the European Research Council (ERC) consolidator grant No.~725978 SYSTEMATICGRAPH.
\end{acks}

\bibliographystyle{ACM-Reference-Format}
\balance
\bibliography{main5}


\begin{thebibliography}{63}


\ifx \showCODEN    \undefined \def \showCODEN     #1{\unskip}     \fi
\ifx \showDOI      \undefined \def \showDOI       #1{#1}\fi
\ifx \showISBNx    \undefined \def \showISBNx     #1{\unskip}     \fi
\ifx \showISBNxiii \undefined \def \showISBNxiii  #1{\unskip}     \fi
\ifx \showISSN     \undefined \def \showISSN      #1{\unskip}     \fi
\ifx \showLCCN     \undefined \def \showLCCN      #1{\unskip}     \fi
\ifx \shownote     \undefined \def \shownote      #1{#1}          \fi
\ifx \showarticletitle \undefined \def \showarticletitle #1{#1}   \fi
\ifx \showURL      \undefined \def \showURL       {\relax}        \fi
\providecommand\bibfield[2]{#2}
\providecommand\bibinfo[2]{#2}
\providecommand\natexlab[1]{#1}
\providecommand\showeprint[2][]{arXiv:#2}

\bibitem[\protect\citeauthoryear{Abboud, Backurs, and Williams}{Abboud
  et~al\mbox{.}}{2015}]%
        {AbboudBW15}
\bibfield{author}{\bibinfo{person}{Amir Abboud}, \bibinfo{person}{Arturs
  Backurs}, {and} \bibinfo{person}{Virginia~Vassilevska Williams}.}
  \bibinfo{year}{2015}\natexlab{}.
\newblock \showarticletitle{Tight Hardness Results for {LCS} and Other Sequence
  Similarity Measures}. In \bibinfo{booktitle}{\emph{Proceedings of the 56th
  Annual {IEEE} Symposium on Foundations of Computer Science ({FOCS} 2015)}}.
  \bibinfo{pages}{59--78}.
\newblock


\bibitem[\protect\citeauthoryear{Abboud, Backurs, and Williams}{Abboud
  et~al\mbox{.}}{2018a}]%
        {DBLP:journals/siamcomp/AbboudBW18}
\bibfield{author}{\bibinfo{person}{Amir Abboud}, \bibinfo{person}{Arturs
  Backurs}, {and} \bibinfo{person}{Virginia~Vassilevska Williams}.}
  \bibinfo{year}{2018}\natexlab{a}.
\newblock \showarticletitle{If the Current Clique Algorithms Are Optimal, so Is
  Valiant's Parser}.
\newblock \bibinfo{journal}{\emph{{SIAM} J. Comput.}} \bibinfo{volume}{47},
  \bibinfo{number}{6} (\bibinfo{year}{2018}), \bibinfo{pages}{2527--2555}.
\newblock
\urldef\tempurl%
\url{https://doi.org/10.1137/16M1061771}
\showDOI{\tempurl}


\bibitem[\protect\citeauthoryear{Abboud, Bringmann, Dell, and Nederlof}{Abboud
  et~al\mbox{.}}{2018b}]%
        {DBLP:conf/stoc/AbboudBDN18}
\bibfield{author}{\bibinfo{person}{Amir Abboud}, \bibinfo{person}{Karl
  Bringmann}, \bibinfo{person}{Holger Dell}, {and} \bibinfo{person}{Jesper
  Nederlof}.} \bibinfo{year}{2018}\natexlab{b}.
\newblock \showarticletitle{More consequences of falsifying {SETH} and the
  orthogonal vectors conjecture}. In \bibinfo{booktitle}{\emph{Proceedings of
  the 50th Annual {ACM} {SIGACT} Symposium on Theory of Computing ({STOC}
  2018), Los Angeles, CA, USA, June 25-29, 2018}},
  \bibfield{editor}{\bibinfo{person}{Ilias Diakonikolas},
  \bibinfo{person}{David Kempe}, {and} \bibinfo{person}{Monika Henzinger}}
  (Eds.). \bibinfo{publisher}{{ACM}}, \bibinfo{pages}{253--266}.
\newblock
\urldef\tempurl%
\url{https://doi.org/10.1145/3188745.3188938}
\showDOI{\tempurl}


\bibitem[\protect\citeauthoryear{Abboud and Williams}{Abboud and
  Williams}{2014}]%
        {DBLP:conf/focs/AbboudW14}
\bibfield{author}{\bibinfo{person}{Amir Abboud} {and}
  \bibinfo{person}{Virginia~Vassilevska Williams}.}
  \bibinfo{year}{2014}\natexlab{}.
\newblock \showarticletitle{Popular Conjectures Imply Strong Lower Bounds for
  Dynamic Problems}. In \bibinfo{booktitle}{\emph{55th {IEEE} Annual Symposium
  on Foundations of Computer Science, {FOCS} 2014, Philadelphia, PA, USA,
  October 18-21, 2014}}. \bibinfo{publisher}{{IEEE} Computer Society},
  \bibinfo{pages}{434--443}.
\newblock
\urldef\tempurl%
\url{https://doi.org/10.1109/FOCS.2014.53}
\showDOI{\tempurl}


\bibitem[\protect\citeauthoryear{Abboud, Williams, and Yu}{Abboud
  et~al\mbox{.}}{2018c}]%
        {DBLP:journals/siamcomp/AbboudWY18}
\bibfield{author}{\bibinfo{person}{Amir Abboud},
  \bibinfo{person}{Virginia~Vassilevska Williams}, {and}
  \bibinfo{person}{Huacheng Yu}.} \bibinfo{year}{2018}\natexlab{c}.
\newblock \showarticletitle{Matching Triangles and Basing Hardness on an
  Extremely Popular Conjecture}.
\newblock \bibinfo{journal}{\emph{{SIAM} J. Comput.}} \bibinfo{volume}{47},
  \bibinfo{number}{3} (\bibinfo{year}{2018}), \bibinfo{pages}{1098--1122}.
\newblock
\urldef\tempurl%
\url{https://doi.org/10.1137/15M1050987}
\showDOI{\tempurl}


\bibitem[\protect\citeauthoryear{Alman and Williams}{Alman and
  Williams}{2021}]%
        {DBLP:conf/soda/AlmanW21}
\bibfield{author}{\bibinfo{person}{Josh Alman} {and}
  \bibinfo{person}{Virginia~Vassilevska Williams}.}
  \bibinfo{year}{2021}\natexlab{}.
\newblock \showarticletitle{A Refined Laser Method and Faster Matrix
  Multiplication}. In \bibinfo{booktitle}{\emph{Proceedings of the 2021
  {ACM-SIAM} Symposium on Discrete Algorithms, {SODA} 2021, Virtual Conference,
  January 10 - 13, 2021}}, \bibfield{editor}{\bibinfo{person}{D{\'{a}}niel
  Marx}} (Ed.). \bibinfo{publisher}{{SIAM}}, \bibinfo{pages}{522--539}.
\newblock
\urldef\tempurl%
\url{https://doi.org/10.1137/1.9781611976465.32}
\showDOI{\tempurl}


\bibitem[\protect\citeauthoryear{Alon, Yuster, and Zwick}{Alon
  et~al\mbox{.}}{1997}]%
        {DBLP:journals/algorithmica/AlonYZ97}
\bibfield{author}{\bibinfo{person}{Noga Alon}, \bibinfo{person}{Raphael
  Yuster}, {and} \bibinfo{person}{Uri Zwick}.} \bibinfo{year}{1997}\natexlab{}.
\newblock \showarticletitle{Finding and Counting Given Length Cycles}.
\newblock \bibinfo{journal}{\emph{Algorithmica}} \bibinfo{volume}{17},
  \bibinfo{number}{3} (\bibinfo{year}{1997}), \bibinfo{pages}{209--223}.
\newblock
\urldef\tempurl%
\url{https://doi.org/10.1007/BF02523189}
\showDOI{\tempurl}


\bibitem[\protect\citeauthoryear{Arora and Barak}{Arora and Barak}{2009}]%
        {DBLP:books/daglib/0023084}
\bibfield{author}{\bibinfo{person}{Sanjeev Arora} {and} \bibinfo{person}{Boaz
  Barak}.} \bibinfo{year}{2009}\natexlab{}.
\newblock \bibinfo{booktitle}{\emph{Computational Complexity - {A} Modern
  Approach}}.
\newblock \bibinfo{publisher}{Cambridge University Press}.
\newblock
\showISBNx{978-0-521-42426-4}
\urldef\tempurl%
\url{http://www.cambridge.org/catalogue/catalogue.asp?isbn=9780521424264}
\showURL{%
\tempurl}


\bibitem[\protect\citeauthoryear{Atserias, Grohe, and Marx}{Atserias
  et~al\mbox{.}}{2013}]%
        {DBLP:journals/siamcomp/AtseriasGM13}
\bibfield{author}{\bibinfo{person}{Albert Atserias}, \bibinfo{person}{Martin
  Grohe}, {and} \bibinfo{person}{D{\'{a}}niel Marx}.}
  \bibinfo{year}{2013}\natexlab{}.
\newblock \showarticletitle{Size Bounds and Query Plans for Relational Joins}.
\newblock \bibinfo{journal}{\emph{{SIAM} J. Comput.}} \bibinfo{volume}{42},
  \bibinfo{number}{4} (\bibinfo{year}{2013}), \bibinfo{pages}{1737--1767}.
\newblock
\urldef\tempurl%
\url{https://doi.org/10.1137/110859440}
\showDOI{\tempurl}


\bibitem[\protect\citeauthoryear{Babai}{Babai}{2016}]%
        {DBLP:conf/stoc/Babai16}
\bibfield{author}{\bibinfo{person}{L{\'{a}}szl{\'{o}} Babai}.}
  \bibinfo{year}{2016}\natexlab{}.
\newblock \showarticletitle{Graph isomorphism in quasipolynomial time [extended
  abstract]}. In \bibinfo{booktitle}{\emph{Proceedings of the 48th Annual {ACM}
  {SIGACT} Symposium on Theory of Computing, {STOC} 2016, Cambridge, MA, USA,
  June 18-21, 2016}}, \bibfield{editor}{\bibinfo{person}{Daniel Wichs} {and}
  \bibinfo{person}{Yishay Mansour}} (Eds.). \bibinfo{publisher}{{ACM}},
  \bibinfo{pages}{684--697}.
\newblock
\urldef\tempurl%
\url{https://doi.org/10.1145/2897518.2897542}
\showDOI{\tempurl}


\bibitem[\protect\citeauthoryear{Babai and Luks}{Babai and Luks}{1983}]%
        {DBLP:conf/stoc/BabaiL83}
\bibfield{author}{\bibinfo{person}{L{\'{a}}szl{\'{o}} Babai} {and}
  \bibinfo{person}{Eugene~M. Luks}.} \bibinfo{year}{1983}\natexlab{}.
\newblock \showarticletitle{Canonical Labeling of Graphs}. In
  \bibinfo{booktitle}{\emph{Proceedings of the 15th Annual {ACM} Symposium on
  Theory of Computing, 25-27 April, 1983, Boston, Massachusetts, {USA}}},
  \bibfield{editor}{\bibinfo{person}{David~S. Johnson}, \bibinfo{person}{Ronald
  Fagin}, \bibinfo{person}{Michael~L. Fredman}, \bibinfo{person}{David Harel},
  \bibinfo{person}{Richard~M. Karp}, \bibinfo{person}{Nancy~A. Lynch},
  \bibinfo{person}{Christos~H. Papadimitriou}, \bibinfo{person}{Ronald~L.
  Rivest}, \bibinfo{person}{Walter~L. Ruzzo}, {and} \bibinfo{person}{Joel~I.
  Seiferas}} (Eds.). \bibinfo{publisher}{{ACM}}, \bibinfo{pages}{171--183}.
\newblock
\urldef\tempurl%
\url{https://doi.org/10.1145/800061.808746}
\showDOI{\tempurl}


\bibitem[\protect\citeauthoryear{Backurs and Indyk}{Backurs and Indyk}{2018}]%
        {DBLP:journals/siamcomp/BackursI18}
\bibfield{author}{\bibinfo{person}{Arturs Backurs} {and} \bibinfo{person}{Piotr
  Indyk}.} \bibinfo{year}{2018}\natexlab{}.
\newblock \showarticletitle{Edit Distance Cannot Be Computed in Strongly
  Subquadratic Time (Unless {SETH} is False)}.
\newblock \bibinfo{journal}{\emph{{SIAM} J. Comput.}} \bibinfo{volume}{47},
  \bibinfo{number}{3} (\bibinfo{year}{2018}), \bibinfo{pages}{1087--1097}.
\newblock
\urldef\tempurl%
\url{https://doi.org/10.1137/15M1053128}
\showDOI{\tempurl}


\bibitem[\protect\citeauthoryear{Bagan, Durand, and Grandjean}{Bagan
  et~al\mbox{.}}{2007}]%
        {DBLP:conf/csl/BaganDG07}
\bibfield{author}{\bibinfo{person}{Guillaume Bagan}, \bibinfo{person}{Arnaud
  Durand}, {and} \bibinfo{person}{Etienne Grandjean}.}
  \bibinfo{year}{2007}\natexlab{}.
\newblock \showarticletitle{On Acyclic Conjunctive Queries and Constant Delay
  Enumeration}. In \bibinfo{booktitle}{\emph{Computer Science Logic, 21st
  International Workshop, {CSL} 2007, 16th Annual Conference of the EACSL,
  Lausanne, Switzerland, September 11-15, 2007, Proceedings}}
  \emph{(\bibinfo{series}{Lecture Notes in Computer Science})},
  \bibfield{editor}{\bibinfo{person}{Jacques Duparc} {and}
  \bibinfo{person}{Thomas~A. Henzinger}} (Eds.), Vol.~\bibinfo{volume}{4646}.
  \bibinfo{publisher}{Springer}, \bibinfo{pages}{208--222}.
\newblock
\urldef\tempurl%
\url{https://doi.org/10.1007/978-3-540-74915-8\_18}
\showDOI{\tempurl}


\bibitem[\protect\citeauthoryear{Bertel{\`{e}} and Brioschi}{Bertel{\`{e}} and
  Brioschi}{1973}]%
        {DBLP:journals/jct/BerteleB73}
\bibfield{author}{\bibinfo{person}{Umberto Bertel{\`{e}}} {and}
  \bibinfo{person}{Francesco Brioschi}.} \bibinfo{year}{1973}\natexlab{}.
\newblock \showarticletitle{On Non-serial Dynamic Programming}.
\newblock \bibinfo{journal}{\emph{J. Comb. Theory, Ser. {A}}}
  \bibinfo{volume}{14}, \bibinfo{number}{2} (\bibinfo{year}{1973}),
  \bibinfo{pages}{137--148}.
\newblock
\urldef\tempurl%
\url{https://doi.org/10.1016/0097-3165(73)90016-2}
\showDOI{\tempurl}


\bibitem[\protect\citeauthoryear{Borradaile and Le}{Borradaile and Le}{2016}]%
        {BorradaileL15}
\bibfield{author}{\bibinfo{person}{Glencora Borradaile} {and}
  \bibinfo{person}{Hung Le}.} \bibinfo{year}{2016}\natexlab{}.
\newblock \showarticletitle{Optimal Dynamic Program for r-Domination Problems
  over Tree Decompositions}. In \bibinfo{booktitle}{\emph{11th International
  Symposium on Parameterized and Exact Computation ({IPEC} 2016)}}
  \emph{(\bibinfo{series}{LIPIcs})}, \bibfield{editor}{\bibinfo{person}{Jiong
  Guo} {and} \bibinfo{person}{Danny Hermelin}} (Eds.),
  Vol.~\bibinfo{volume}{63}. \bibinfo{publisher}{Schloss Dagstuhl -
  Leibniz-Zentrum fuer Informatik}, \bibinfo{pages}{8:1--8:23}.
\newblock
\urldef\tempurl%
\url{https://doi.org/10.4230/LIPIcs.IPEC.2016.8}
\showDOI{\tempurl}


\bibitem[\protect\citeauthoryear{Brault{-}Baron}{Brault{-}Baron}{2013}]%
        {DBLP:phd/hal/BraultBaron13}
\bibfield{author}{\bibinfo{person}{Johann Brault{-}Baron}.}
  \bibinfo{year}{2013}\natexlab{}.
\newblock \emph{\bibinfo{title}{De la pertinence de l'{\'{e}}num{\'{e}}ration :
  complexit{\'{e}} en logiques propositionnelle et du premier ordre. (The
  relevance of the list: propositional logic and complexity of the first
  order)}}.
\newblock \bibinfo{thesistype}{Ph.D. Dissertation}. \bibinfo{school}{University
  of Caen Normandy, France}.
\newblock
\urldef\tempurl%
\url{https://tel.archives-ouvertes.fr/tel-01081392}
\showURL{%
\tempurl}


\bibitem[\protect\citeauthoryear{Bringmann}{Bringmann}{2014}]%
        {DBLP:conf/focs/Bringmann14}
\bibfield{author}{\bibinfo{person}{Karl Bringmann}.}
  \bibinfo{year}{2014}\natexlab{}.
\newblock \showarticletitle{Why Walking the Dog Takes Time: Frechet Distance
  Has No Strongly Subquadratic Algorithms Unless {SETH} Fails}. In
  \bibinfo{booktitle}{\emph{55th {IEEE} Annual Symposium on Foundations of
  Computer Science ({FOCS} 2014), Philadelphia, PA, USA, October 18-21, 2014}}.
  \bibinfo{publisher}{{IEEE} Computer Society}, \bibinfo{pages}{661--670}.
\newblock
\urldef\tempurl%
\url{https://doi.org/10.1109/FOCS.2014.76}
\showDOI{\tempurl}


\bibitem[\protect\citeauthoryear{Bringmann, Gawrychowski, Mozes, and
  Weimann}{Bringmann et~al\mbox{.}}{2018}]%
        {DBLP:conf/soda/BringmannGMW18}
\bibfield{author}{\bibinfo{person}{Karl Bringmann}, \bibinfo{person}{Pawel
  Gawrychowski}, \bibinfo{person}{Shay Mozes}, {and} \bibinfo{person}{Oren
  Weimann}.} \bibinfo{year}{2018}\natexlab{}.
\newblock \showarticletitle{Tree Edit Distance Cannot be Computed in Strongly
  Subcubic Time (unless {APSP} can)}. In \bibinfo{booktitle}{\emph{Proceedings
  of the Twenty-Ninth Annual {ACM-SIAM} Symposium on Discrete Algorithms
  ({SODA} 2018)}}, \bibfield{editor}{\bibinfo{person}{Artur Czumaj}} (Ed.).
  \bibinfo{publisher}{{SIAM}}, \bibinfo{pages}{1190--1206}.
\newblock
\urldef\tempurl%
\url{https://doi.org/10.1137/1.9781611975031.77}
\showDOI{\tempurl}


\bibitem[\protect\citeauthoryear{Bringmann and K{\"{u}}nnemann}{Bringmann and
  K{\"{u}}nnemann}{2015}]%
        {DBLP:conf/focs/BringmannK15}
\bibfield{author}{\bibinfo{person}{Karl Bringmann} {and}
  \bibinfo{person}{Marvin K{\"{u}}nnemann}.} \bibinfo{year}{2015}\natexlab{}.
\newblock \showarticletitle{Quadratic Conditional Lower Bounds for String
  Problems and Dynamic Time Warping}. In \bibinfo{booktitle}{\emph{{IEEE} 56th
  Annual Symposium on Foundations of Computer Science ({FOCS} 2015), Berkeley,
  CA, USA, 17-20 October, 2015}}, \bibfield{editor}{\bibinfo{person}{Venkatesan
  Guruswami}} (Ed.). \bibinfo{publisher}{{IEEE} Computer Society},
  \bibinfo{pages}{79--97}.
\newblock
\urldef\tempurl%
\url{https://doi.org/10.1109/FOCS.2015.15}
\showDOI{\tempurl}


\bibitem[\protect\citeauthoryear{Bringmann and K{\"{u}}nnemann}{Bringmann and
  K{\"{u}}nnemann}{2018}]%
        {DBLP:conf/soda/BringmannK18}
\bibfield{author}{\bibinfo{person}{Karl Bringmann} {and}
  \bibinfo{person}{Marvin K{\"{u}}nnemann}.} \bibinfo{year}{2018}\natexlab{}.
\newblock \showarticletitle{Multivariate Fine-Grained Complexity of Longest
  Common Subsequence}. In \bibinfo{booktitle}{\emph{Proceedings of the
  Twenty-Ninth Annual {ACM-SIAM} Symposium on Discrete Algorithms ({SODA}
  2018)}}, \bibfield{editor}{\bibinfo{person}{Artur Czumaj}} (Ed.).
  \bibinfo{publisher}{{SIAM}}, \bibinfo{pages}{1216--1235}.
\newblock
\urldef\tempurl%
\url{https://doi.org/10.1137/1.9781611975031.79}
\showDOI{\tempurl}


\bibitem[\protect\citeauthoryear{Bulatov}{Bulatov}{2017}]%
        {DBLP:conf/focs/Bulatov17}
\bibfield{author}{\bibinfo{person}{Andrei~A. Bulatov}.}
  \bibinfo{year}{2017}\natexlab{}.
\newblock \showarticletitle{A Dichotomy Theorem for Nonuniform {C}{S}{P}s}. In
  \bibinfo{booktitle}{\emph{58th {IEEE} Annual Symposium on Foundations of
  Computer Science ({FOCS} 2017), Berkeley, CA, USA, October 15-17, 2017}},
  \bibfield{editor}{\bibinfo{person}{Chris Umans}} (Ed.).
  \bibinfo{publisher}{{IEEE} Computer Society}, \bibinfo{pages}{319--330}.
\newblock
\urldef\tempurl%
\url{https://doi.org/10.1109/FOCS.2017.37}
\showDOI{\tempurl}


\bibitem[\protect\citeauthoryear{Chen, Huang, Kanj, and Xia}{Chen
  et~al\mbox{.}}{2004}]%
        {MR2121603}
\bibfield{author}{\bibinfo{person}{Jianer Chen}, \bibinfo{person}{Xiuzhen
  Huang}, \bibinfo{person}{Iyad~A. Kanj}, {and} \bibinfo{person}{Ge Xia}.}
  \bibinfo{year}{2004}\natexlab{}.
\newblock \showarticletitle{Linear {FPT} reductions and computational lower
  bounds}. In \bibinfo{booktitle}{\emph{Proceedings of the 36th Annual ACM
  Symposium on Theory of Computing (STOC 2004)}}. \bibinfo{publisher}{ACM},
  \bibinfo{address}{New York}, \bibinfo{pages}{212--221}.
\newblock


\bibitem[\protect\citeauthoryear{Chen, Kanj, and Jia}{Chen
  et~al\mbox{.}}{2001}]%
        {MR2002h:05149}
\bibfield{author}{\bibinfo{person}{Jianer Chen}, \bibinfo{person}{Iyad~A.
  Kanj}, {and} \bibinfo{person}{Weijia Jia}.} \bibinfo{year}{2001}\natexlab{}.
\newblock \showarticletitle{Vertex cover: further observations and further
  improvements}.
\newblock \bibinfo{journal}{\emph{J. Algorithms}} \bibinfo{volume}{41},
  \bibinfo{number}{2} (\bibinfo{year}{2001}), \bibinfo{pages}{280--301}.
\newblock
\showCODEN{JOALDV}
\showISSN{0196-6774}


\bibitem[\protect\citeauthoryear{Chung, Graham, Frankl, and Shearer}{Chung
  et~al\mbox{.}}{1986}]%
        {MR859293}
\bibfield{author}{\bibinfo{person}{F.~R.~K. Chung}, \bibinfo{person}{R.~L.
  Graham}, \bibinfo{person}{P. Frankl}, {and} \bibinfo{person}{J.~B. Shearer}.}
  \bibinfo{year}{1986}\natexlab{}.
\newblock \showarticletitle{Some intersection theorems for ordered sets and
  graphs}.
\newblock \bibinfo{journal}{\emph{J. Combin. Theory Ser. A}}
  \bibinfo{volume}{43}, \bibinfo{number}{1} (\bibinfo{year}{1986}),
  \bibinfo{pages}{23--37}.
\newblock
\showCODEN{JCBTA7}
\showISSN{0097-3165}


\bibitem[\protect\citeauthoryear{Cohen{-}Addad, de~Verdi{\`{e}}re, Marx, and
  de~Mesmay}{Cohen{-}Addad et~al\mbox{.}}{[n.d.]}]%
        {jacm-genus}
\bibfield{author}{\bibinfo{person}{Vincent Cohen{-}Addad},
  \bibinfo{person}{{\'{E}}ric~Colin de Verdi{\`{e}}re},
  \bibinfo{person}{D{\'{a}}niel Marx}, {and} \bibinfo{person}{Arnaud de
  Mesmay}.} \bibinfo{year}{[n.d.]}\natexlab{}.
\newblock \bibinfo{title}{Almost Tight Lower Bounds for Hard Cutting Problems
  in Embedded Graphs}.
\newblock
\newblock
\newblock
\shownote{To appear in {\em Journal of the ACM.}}


\bibitem[\protect\citeauthoryear{Cook}{Cook}{1971}]%
        {DBLP:conf/stoc/Cook71}
\bibfield{author}{\bibinfo{person}{Stephen~A. Cook}.}
  \bibinfo{year}{1971}\natexlab{}.
\newblock \showarticletitle{The Complexity of Theorem-Proving Procedures}. In
  \bibinfo{booktitle}{\emph{Proceedings of the 3rd Annual {ACM} Symposium on
  Theory of Computing, May 3-5, 1971, Shaker Heights, Ohio, {USA}}},
  \bibfield{editor}{\bibinfo{person}{Michael~A. Harrison},
  \bibinfo{person}{Ranan~B. Banerji}, {and} \bibinfo{person}{Jeffrey~D.
  Ullman}} (Eds.). \bibinfo{publisher}{{ACM}}, \bibinfo{pages}{151--158}.
\newblock
\urldef\tempurl%
\url{https://doi.org/10.1145/800157.805047}
\showDOI{\tempurl}


\bibitem[\protect\citeauthoryear{Curticapean and Marx}{Curticapean and
  Marx}{2016}]%
        {CurticapeanM16}
\bibfield{author}{\bibinfo{person}{Radu Curticapean} {and}
  \bibinfo{person}{D{\'{a}}niel Marx}.} \bibinfo{year}{2016}\natexlab{}.
\newblock \showarticletitle{Tight conditional lower bounds for counting perfect
  matchings on graphs of bounded treewidth, cliquewidth, and genus}. In
  \bibinfo{booktitle}{\emph{Proceedings of the 27th Annual {ACM-SIAM} Symposium
  on Discrete Algorithms ({SODA} 2016)}}. \bibinfo{pages}{1650--1669}.
\newblock


\bibitem[\protect\citeauthoryear{Cygan, Dell, Lokshtanov, Marx, Nederlof,
  Okamoto, Paturi, Saurabh, and Wahlstr{\"{o}}m}{Cygan et~al\mbox{.}}{2016}]%
        {CyganDLMNOPSW16}
\bibfield{author}{\bibinfo{person}{Marek Cygan}, \bibinfo{person}{Holger Dell},
  \bibinfo{person}{Daniel Lokshtanov}, \bibinfo{person}{D{\'{a}}niel Marx},
  \bibinfo{person}{Jesper Nederlof}, \bibinfo{person}{Yoshio Okamoto},
  \bibinfo{person}{Ramamohan Paturi}, \bibinfo{person}{Saket Saurabh}, {and}
  \bibinfo{person}{Magnus Wahlstr{\"{o}}m}.} \bibinfo{year}{2016}\natexlab{}.
\newblock \showarticletitle{On Problems as Hard as {CNF-SAT}}.
\newblock \bibinfo{journal}{\emph{{ACM} Trans. Algorithms}}
  \bibinfo{volume}{12}, \bibinfo{number}{3} (\bibinfo{year}{2016}),
  \bibinfo{pages}{41:1--41:24}.
\newblock


\bibitem[\protect\citeauthoryear{Cygan, Fomin, Kowalik, Lokshtanov, Marx,
  Pilipczuk, Pilipczuk, and Saurabh}{Cygan et~al\mbox{.}}{2015}]%
        {CyganFKLMPPS15}
\bibfield{author}{\bibinfo{person}{Marek Cygan}, \bibinfo{person}{Fedor~V.
  Fomin}, \bibinfo{person}{Lukasz Kowalik}, \bibinfo{person}{Daniel
  Lokshtanov}, \bibinfo{person}{D{\'{a}}niel Marx}, \bibinfo{person}{Marcin
  Pilipczuk}, \bibinfo{person}{Michal Pilipczuk}, {and} \bibinfo{person}{Saket
  Saurabh}.} \bibinfo{year}{2015}\natexlab{}.
\newblock \bibinfo{booktitle}{\emph{Parameterized Algorithms}}.
\newblock \bibinfo{publisher}{Springer}.
\newblock


\bibitem[\protect\citeauthoryear{Cygan, Nederlof, Pilipczuk, Pilipczuk, van
  Rooij, and Wojtaszczyk}{Cygan et~al\mbox{.}}{2011}]%
        {CyganNPPRW11}
\bibfield{author}{\bibinfo{person}{Marek Cygan}, \bibinfo{person}{Jesper
  Nederlof}, \bibinfo{person}{Marcin Pilipczuk}, \bibinfo{person}{Michal
  Pilipczuk}, \bibinfo{person}{Johan M.~M. van Rooij}, {and}
  \bibinfo{person}{Jakub~Onufry Wojtaszczyk}.} \bibinfo{year}{2011}\natexlab{}.
\newblock \showarticletitle{Solving Connectivity Problems Parameterized by
  Treewidth in Single Exponential Time}. In
  \bibinfo{booktitle}{\emph{Proceedings of the 52nd Annual {IEEE} Symposium on
  Foundations of Computer Science ({FOCS} 2011)}}. \bibinfo{pages}{150--159}.
\newblock


\bibitem[\protect\citeauthoryear{Downey and Fellows}{Downey and
  Fellows}{1999}]%
        {MR2001b:68042}
\bibfield{author}{\bibinfo{person}{Rodney~G. Downey} {and}
  \bibinfo{person}{Michael~R. Fellows}.} \bibinfo{year}{1999}\natexlab{}.
\newblock \bibinfo{booktitle}{\emph{Parameterized Complexity}}.
\newblock \bibinfo{publisher}{Springer}, \bibinfo{address}{New York}. xvi+533
  pages.
\newblock
\showISBNx{0-387-94883-X}


\bibitem[\protect\citeauthoryear{Downey and Fellows}{Downey and
  Fellows}{2013}]%
        {DBLP:series/txcs/DowneyF13}
\bibfield{author}{\bibinfo{person}{Rodney~G. Downey} {and}
  \bibinfo{person}{Michael~R. Fellows}.} \bibinfo{year}{2013}\natexlab{}.
\newblock \bibinfo{booktitle}{\emph{Fundamentals of Parameterized Complexity}}.
\newblock \bibinfo{publisher}{Springer}.
\newblock
\showISBNx{978-1-4471-5558-4}
\urldef\tempurl%
\url{https://doi.org/10.1007/978-1-4471-5559-1}
\showDOI{\tempurl}


\bibitem[\protect\citeauthoryear{Egri, Marx, and Rz{\k{a}}{\.{z}}ewski}{Egri
  et~al\mbox{.}}{2018}]%
        {DBLP:conf/stacs/EgriMR18}
\bibfield{author}{\bibinfo{person}{L{\'{a}}szl{\'{o}} Egri},
  \bibinfo{person}{D{\'{a}}niel Marx}, {and} \bibinfo{person}{Pawe{\l}
  Rz{\k{a}}{\.{z}}ewski}.} \bibinfo{year}{2018}\natexlab{}.
\newblock \showarticletitle{Finding List Homomorphisms from Bounded-treewidth
  Graphs to Reflexive Graphs: a Complete Complexity Characterization}. In
  \bibinfo{booktitle}{\emph{35th Symposium on Theoretical Aspects of Computer
  Science (STACS 2018)}}. \bibinfo{pages}{27:1--27:15}.
\newblock
\urldef\tempurl%
\url{https://doi.org/10.4230/LIPIcs.STACS.2018.27}
\showDOI{\tempurl}


\bibitem[\protect\citeauthoryear{Eisenbrand and Grandoni}{Eisenbrand and
  Grandoni}{2004}]%
        {DBLP:journals/tcs/EisenbrandG04}
\bibfield{author}{\bibinfo{person}{Friedrich Eisenbrand} {and}
  \bibinfo{person}{Fabrizio Grandoni}.} \bibinfo{year}{2004}\natexlab{}.
\newblock \showarticletitle{On the complexity of fixed parameter clique and
  dominating set}.
\newblock \bibinfo{journal}{\emph{Theor. Comput. Sci.}} \bibinfo{volume}{326},
  \bibinfo{number}{1-3} (\bibinfo{year}{2004}), \bibinfo{pages}{57--67}.
\newblock
\urldef\tempurl%
\url{https://doi.org/10.1016/j.tcs.2004.05.009}
\showDOI{\tempurl}


\bibitem[\protect\citeauthoryear{Feder and Vardi}{Feder and Vardi}{1999}]%
        {MR2000e:68063}
\bibfield{author}{\bibinfo{person}{Tom{\'a}s Feder} {and}
  \bibinfo{person}{Moshe~Y. Vardi}.} \bibinfo{year}{1999}\natexlab{}.
\newblock \showarticletitle{The computational structure of monotone monadic
  {SNP} and constraint satisfaction: a study through {D}atalog and group
  theory}.
\newblock \bibinfo{journal}{\emph{SIAM J. Comput.}} \bibinfo{volume}{28},
  \bibinfo{number}{1} (\bibinfo{year}{1999}), \bibinfo{pages}{57--104}.
\newblock
\showISSN{1095-7111}


\bibitem[\protect\citeauthoryear{Flum and Grohe}{Flum and Grohe}{2006}]%
        {MR2238686}
\bibfield{author}{\bibinfo{person}{J{\"o}rg Flum} {and} \bibinfo{person}{Martin
  Grohe}.} \bibinfo{year}{2006}\natexlab{}.
\newblock \bibinfo{booktitle}{\emph{Parameterized Complexity Theory}}.
\newblock \bibinfo{publisher}{Springer}, \bibinfo{address}{Berlin}. xiv+493
  pages.
\newblock
\showISBNx{978-3-540-29952-3; 3-540-29952-1}


\bibitem[\protect\citeauthoryear{Freuder}{Freuder}{1990}]%
        {Freuder90AA}
\bibfield{author}{\bibinfo{person}{E.~C. Freuder}.}
  \bibinfo{year}{1990}\natexlab{}.
\newblock \showarticletitle{Complexity of K-Tree Structured Constraint
  Satisfaction Problems}. In \bibinfo{booktitle}{\emph{Proc. of AAAI-90}}.
  \bibinfo{address}{Boston, MA}, \bibinfo{pages}{4--9}.
\newblock


\bibitem[\protect\citeauthoryear{Garey and Johnson}{Garey and Johnson}{1979}]%
        {DBLP:books/fm/GareyJ79}
\bibfield{author}{\bibinfo{person}{M.~R. Garey} {and} \bibinfo{person}{David~S.
  Johnson}.} \bibinfo{year}{1979}\natexlab{}.
\newblock \bibinfo{booktitle}{\emph{Computers and Intractability: {A} Guide to
  the Theory of NP-Completeness}}.
\newblock \bibinfo{publisher}{W. H. Freeman}.
\newblock
\showISBNx{0-7167-1044-7}


\bibitem[\protect\citeauthoryear{Grohe}{Grohe}{2007}]%
        {1206036}
\bibfield{author}{\bibinfo{person}{Martin Grohe}.}
  \bibinfo{year}{2007}\natexlab{}.
\newblock \showarticletitle{The complexity of homomorphism and constraint
  satisfaction problems seen from the other side}.
\newblock \bibinfo{journal}{\emph{J. ACM}} \bibinfo{volume}{54},
  \bibinfo{number}{1} (\bibinfo{year}{2007}), \bibinfo{pages}{1}.
\newblock
\showISSN{0004-5411}
\urldef\tempurl%
\url{https://doi.org/10.1145/1206035.1206036}
\showDOI{\tempurl}


\bibitem[\protect\citeauthoryear{Grohe, Schwentick, and Segoufin}{Grohe
  et~al\mbox{.}}{2001}]%
        {380867}
\bibfield{author}{\bibinfo{person}{Martin Grohe}, \bibinfo{person}{Thomas
  Schwentick}, {and} \bibinfo{person}{Luc Segoufin}.}
  \bibinfo{year}{2001}\natexlab{}.
\newblock \showarticletitle{When is the evaluation of conjunctive queries
  tractable?}. In \bibinfo{booktitle}{\emph{Proceedings of the thirty-third
  annual ACM symposium on Theory of computing (STOC 2001)}} (Hersonissos,
  Greece). \bibinfo{publisher}{ACM Press}, \bibinfo{address}{New York, NY,
  USA}, \bibinfo{pages}{657--666}.
\newblock
\showISBNx{1-58113-349-9}
\urldef\tempurl%
\url{https://doi.org/10.1145/380752.380867}
\showDOI{\tempurl}


\bibitem[\protect\citeauthoryear{Halin}{Halin}{1976}]%
        {halin72}
\bibfield{author}{\bibinfo{person}{Rudolf Halin}.}
  \bibinfo{year}{1976}\natexlab{}.
\newblock \showarticletitle{S-functions for graphs}.
\newblock \bibinfo{journal}{\emph{Journal of Geometry}} \bibinfo{volume}{8},
  \bibinfo{number}{1-2} (\bibinfo{year}{1976}), \bibinfo{pages}{171--186}.
\newblock


\bibitem[\protect\citeauthoryear{Hertli}{Hertli}{2014}]%
        {DBLP:journals/siamcomp/Hertli14}
\bibfield{author}{\bibinfo{person}{Timon Hertli}.}
  \bibinfo{year}{2014}\natexlab{}.
\newblock \showarticletitle{3-{S}{A}{T} Faster and Simpler ---
  {U}nique-{S}{A}{T} Bounds for {PPSZ} Hold in General}.
\newblock \bibinfo{journal}{\emph{{SIAM} J. Comput.}} \bibinfo{volume}{43},
  \bibinfo{number}{2} (\bibinfo{year}{2014}), \bibinfo{pages}{718--729}.
\newblock
\urldef\tempurl%
\url{https://doi.org/10.1137/120868177}
\showDOI{\tempurl}


\bibitem[\protect\citeauthoryear{Impagliazzo and Paturi}{Impagliazzo and
  Paturi}{2001}]%
        {ImpagliazzoP01}
\bibfield{author}{\bibinfo{person}{Russell Impagliazzo} {and}
  \bibinfo{person}{Ramamohan Paturi}.} \bibinfo{year}{2001}\natexlab{}.
\newblock \showarticletitle{On the Complexity of {$k$}-{S}{A}{T}}.
\newblock \bibinfo{journal}{\emph{J. Comput. Syst. Sci.}} \bibinfo{volume}{62},
  \bibinfo{number}{2} (\bibinfo{year}{2001}), \bibinfo{pages}{367--375}.
\newblock


\bibitem[\protect\citeauthoryear{Impagliazzo, Paturi, and Zane}{Impagliazzo
  et~al\mbox{.}}{2001}]%
        {MR1894519}
\bibfield{author}{\bibinfo{person}{Russell Impagliazzo},
  \bibinfo{person}{Ramamohan Paturi}, {and} \bibinfo{person}{Francis Zane}.}
  \bibinfo{year}{2001}\natexlab{}.
\newblock \showarticletitle{Which problems have strongly exponential
  complexity?}
\newblock \bibinfo{journal}{\emph{J. Comput. System Sci.}}
  \bibinfo{volume}{63}, \bibinfo{number}{4} (\bibinfo{year}{2001}),
  \bibinfo{pages}{512--530}.
\newblock
\showCODEN{JCSSBM}
\showISSN{0022-0000}


\bibitem[\protect\citeauthoryear{Iwata and Yoshida}{Iwata and Yoshida}{2015}]%
        {DBLP:conf/esa/IwataY15}
\bibfield{author}{\bibinfo{person}{Yoichi Iwata} {and} \bibinfo{person}{Yuichi
  Yoshida}.} \bibinfo{year}{2015}\natexlab{}.
\newblock \showarticletitle{On the Equivalence among Problems of Bounded
  Width}. In \bibinfo{booktitle}{\emph{23rd Annual European Symposium (ESA
  2015)}} \emph{(\bibinfo{series}{Lecture Notes in Computer Science})},
  \bibfield{editor}{\bibinfo{person}{Nikhil Bansal} {and}
  \bibinfo{person}{Irene Finocchi}} (Eds.), Vol.~\bibinfo{volume}{9294}.
  \bibinfo{publisher}{Springer}, \bibinfo{pages}{754--765}.
\newblock
\urldef\tempurl%
\url{https://doi.org/10.1007/978-3-662-48350-3\_63}
\showDOI{\tempurl}


\bibitem[\protect\citeauthoryear{Jaffke and Jansen}{Jaffke and Jansen}{2017}]%
        {JaffkeJ17}
\bibfield{author}{\bibinfo{person}{Lars Jaffke} {and} \bibinfo{person}{Bart
  M.~P. Jansen}.} \bibinfo{year}{2017}\natexlab{}.
\newblock \showarticletitle{Fine-Grained Parameterized Complexity Analysis of
  Graph Coloring Problems}. In \bibinfo{booktitle}{\emph{Proceedings of the
  10th International Conference on Algorithms and Complexity (CIAC 2017)}}
  \emph{(\bibinfo{series}{Lecture Notes in Computer Science})},
  Vol.~\bibinfo{volume}{10236}. \bibinfo{pages}{345--356}.
\newblock


\bibitem[\protect\citeauthoryear{Karp}{Karp}{1972}]%
        {DBLP:conf/coco/Karp72}
\bibfield{author}{\bibinfo{person}{Richard~M. Karp}.}
  \bibinfo{year}{1972}\natexlab{}.
\newblock \showarticletitle{Reducibility Among Combinatorial Problems}. In
  \bibinfo{booktitle}{\emph{Proceedings of a symposium on the Complexity of
  Computer Computations, held March 20-22, 1972, at the {IBM} Thomas J. Watson
  Research Center, Yorktown Heights, New York, {USA}}}
  \emph{(\bibinfo{series}{The {IBM} Research Symposia Series})},
  \bibfield{editor}{\bibinfo{person}{Raymond~E. Miller} {and}
  \bibinfo{person}{James~W. Thatcher}} (Eds.). \bibinfo{publisher}{Plenum
  Press, New York}, \bibinfo{pages}{85--103}.
\newblock
\urldef\tempurl%
\url{https://doi.org/10.1007/978-1-4684-2001-2\_9}
\showDOI{\tempurl}


\bibitem[\protect\citeauthoryear{Ladner}{Ladner}{1975}]%
        {DBLP:journals/jacm/Ladner75}
\bibfield{author}{\bibinfo{person}{Richard~E. Ladner}.}
  \bibinfo{year}{1975}\natexlab{}.
\newblock \showarticletitle{On the Structure of Polynomial Time Reducibility}.
\newblock \bibinfo{journal}{\emph{J. {ACM}}} \bibinfo{volume}{22},
  \bibinfo{number}{1} (\bibinfo{year}{1975}), \bibinfo{pages}{155--171}.
\newblock
\urldef\tempurl%
\url{https://doi.org/10.1145/321864.321877}
\showDOI{\tempurl}


\bibitem[\protect\citeauthoryear{Lenstra and Pomerance}{Lenstra and
  Pomerance}{1992}]%
        {MR1137100}
\bibfield{author}{\bibinfo{person}{H.~W. Lenstra, Jr.} {and}
  \bibinfo{person}{Carl Pomerance}.} \bibinfo{year}{1992}\natexlab{}.
\newblock \showarticletitle{A rigorous time bound for factoring integers}.
\newblock \bibinfo{journal}{\emph{J. Amer. Math. Soc.}} \bibinfo{volume}{5},
  \bibinfo{number}{3} (\bibinfo{year}{1992}), \bibinfo{pages}{483--516}.
\newblock
\showISSN{0894-0347}
\urldef\tempurl%
\url{https://doi.org/10.2307/2152702}
\showDOI{\tempurl}


\bibitem[\protect\citeauthoryear{Lincoln, Williams, and Williams}{Lincoln
  et~al\mbox{.}}{2018}]%
        {DBLP:conf/soda/LincolnWW18}
\bibfield{author}{\bibinfo{person}{Andrea Lincoln},
  \bibinfo{person}{Virginia~Vassilevska Williams}, {and}
  \bibinfo{person}{R.~Ryan Williams}.} \bibinfo{year}{2018}\natexlab{}.
\newblock \showarticletitle{Tight Hardness for Shortest Cycles and Paths in
  Sparse Graphs}. In \bibinfo{booktitle}{\emph{Proceedings of the Twenty-Ninth
  Annual {ACM-SIAM} Symposium on Discrete Algorithms, {SODA} 2018, New Orleans,
  LA, USA, January 7-10, 2018}}, \bibfield{editor}{\bibinfo{person}{Artur
  Czumaj}} (Ed.). \bibinfo{publisher}{{SIAM}}, \bibinfo{pages}{1236--1252}.
\newblock
\urldef\tempurl%
\url{https://doi.org/10.1137/1.9781611975031.80}
\showDOI{\tempurl}


\bibitem[\protect\citeauthoryear{Lokshtanov, Marx, and Saurabh}{Lokshtanov
  et~al\mbox{.}}{2018}]%
        {DBLP:journals/talg/LokshtanovMS18}
\bibfield{author}{\bibinfo{person}{Daniel Lokshtanov},
  \bibinfo{person}{D{\'{a}}niel Marx}, {and} \bibinfo{person}{Saket Saurabh}.}
  \bibinfo{year}{2018}\natexlab{}.
\newblock \showarticletitle{Known Algorithms on Graphs of Bounded Treewidth Are
  Probably Optimal}.
\newblock \bibinfo{journal}{\emph{{ACM} Trans. Algorithms}}
  \bibinfo{volume}{14}, \bibinfo{number}{2} (\bibinfo{year}{2018}),
  \bibinfo{pages}{13:1--13:30}.
\newblock
\urldef\tempurl%
\url{https://doi.org/10.1145/3170442}
\showDOI{\tempurl}


\bibitem[\protect\citeauthoryear{Marx}{Marx}{2010}]%
        {marx-toc-treewidth}
\bibfield{author}{\bibinfo{person}{D{\'a}niel Marx}.}
  \bibinfo{year}{2010}\natexlab{}.
\newblock \showarticletitle{Can You Beat Treewidth?}
\newblock \bibinfo{journal}{\emph{Theory of Computing}} \bibinfo{volume}{6},
  \bibinfo{number}{1} (\bibinfo{year}{2010}), \bibinfo{pages}{85--112}.
\newblock
\urldef\tempurl%
\url{https://doi.org/10.4086/toc.2010.v006a005}
\showDOI{\tempurl}
\showeprint{toc:v006/a005}


\bibitem[\protect\citeauthoryear{Ne{\v{s}}et{\v{r}}il and
  Poljak}{Ne{\v{s}}et{\v{r}}il and Poljak}{1985}]%
        {nevsetvril1985complexity}
\bibfield{author}{\bibinfo{person}{Jaroslav Ne{\v{s}}et{\v{r}}il} {and}
  \bibinfo{person}{Svatopluk Poljak}.} \bibinfo{year}{1985}\natexlab{}.
\newblock \showarticletitle{On the complexity of the subgraph problem}.
\newblock \bibinfo{journal}{\emph{Commentationes Mathematicae Universitatis
  Carolinae}} \bibinfo{volume}{26}, \bibinfo{number}{2} (\bibinfo{year}{1985}),
  \bibinfo{pages}{415--419}.
\newblock


\bibitem[\protect\citeauthoryear{Ngo, Porat, R{\'{e}}, and Rudra}{Ngo
  et~al\mbox{.}}{2018}]%
        {DBLP:journals/jacm/NgoPRR18}
\bibfield{author}{\bibinfo{person}{Hung~Q. Ngo}, \bibinfo{person}{Ely Porat},
  \bibinfo{person}{Christopher R{\'{e}}}, {and} \bibinfo{person}{Atri Rudra}.}
  \bibinfo{year}{2018}\natexlab{}.
\newblock \showarticletitle{Worst-case Optimal Join Algorithms}.
\newblock \bibinfo{journal}{\emph{J. {ACM}}} \bibinfo{volume}{65},
  \bibinfo{number}{3} (\bibinfo{year}{2018}), \bibinfo{pages}{16:1--16:40}.
\newblock
\urldef\tempurl%
\url{https://doi.org/10.1145/3180143}
\showDOI{\tempurl}


\bibitem[\protect\citeauthoryear{Papadimitriou}{Papadimitriou}{1994}]%
        {PapadimitriouBook}
\bibfield{author}{\bibinfo{person}{C.~H. Papadimitriou}.}
  \bibinfo{year}{1994}\natexlab{}.
\newblock \bibinfo{booktitle}{\emph{Computational Complexity}}.
\newblock \bibinfo{publisher}{Addison Wesley}.
\newblock


\bibitem[\protect\citeauthoryear{Patrascu and Williams}{Patrascu and
  Williams}{2010}]%
        {patrascu10sat-lbs}
\bibfield{author}{\bibinfo{person}{Mihai Patrascu} {and} \bibinfo{person}{Ryan
  Williams}.} \bibinfo{year}{2010}\natexlab{}.
\newblock \showarticletitle{On the Possibility of Faster {SAT} Algorithms}. In
  \bibinfo{booktitle}{\emph{Proceedings of the 21st Annual {ACM-SIAM} Symposium
  on Discrete Algorithms ({SODA} 2010)}}. \bibinfo{pages}{1065--1075}.
\newblock


\bibitem[\protect\citeauthoryear{Robertson and Seymour}{Robertson and
  Seymour}{1984}]%
        {DBLP:journals/jct/RobertsonS84}
\bibfield{author}{\bibinfo{person}{Neil Robertson} {and}
  \bibinfo{person}{Paul~D. Seymour}.} \bibinfo{year}{1984}\natexlab{}.
\newblock \showarticletitle{Graph minors. {III.} {P}lanar tree-width}.
\newblock \bibinfo{journal}{\emph{J. Comb. Theory, Ser. {B}}}
  \bibinfo{volume}{36}, \bibinfo{number}{1} (\bibinfo{year}{1984}),
  \bibinfo{pages}{49--64}.
\newblock
\urldef\tempurl%
\url{https://doi.org/10.1016/0095-8956(84)90013-3}
\showDOI{\tempurl}


\bibitem[\protect\citeauthoryear{Roditty and Williams}{Roditty and
  Williams}{2013}]%
        {RodittyW13}
\bibfield{author}{\bibinfo{person}{Liam Roditty} {and}
  \bibinfo{person}{Virginia~Vassilevska Williams}.}
  \bibinfo{year}{2013}\natexlab{}.
\newblock \showarticletitle{Fast approximation algorithms for the diameter and
  radius of sparse graphs}. In \bibinfo{booktitle}{\emph{Proceedings of the
  45th Annual {ACM} on Symposium on Theory of Computing ({STOC} 2013)}}.
  \bibinfo{pages}{515--524}.
\newblock


\bibitem[\protect\citeauthoryear{Schaefer}{Schaefer}{1978}]%
        {DBLP:conf/stoc/Schaefer78}
\bibfield{author}{\bibinfo{person}{Thomas~J. Schaefer}.}
  \bibinfo{year}{1978}\natexlab{}.
\newblock \showarticletitle{The Complexity of Satisfiability Problems}. In
  \bibinfo{booktitle}{\emph{Proceedings of the 10th Annual {ACM} Symposium on
  Theory of Computing, May 1-3, 1978, San Diego, California, {USA}}},
  \bibfield{editor}{\bibinfo{person}{Richard~J. Lipton},
  \bibinfo{person}{Walter~A. Burkhard}, \bibinfo{person}{Walter~J. Savitch},
  \bibinfo{person}{Emily~P. Friedman}, {and} \bibinfo{person}{Alfred~V. Aho}}
  (Eds.). \bibinfo{publisher}{{ACM}}, \bibinfo{pages}{216--226}.
\newblock
\urldef\tempurl%
\url{https://doi.org/10.1145/800133.804350}
\showDOI{\tempurl}


\bibitem[\protect\citeauthoryear{Valiant}{Valiant}{1975}]%
        {DBLP:journals/jcss/Valiant75}
\bibfield{author}{\bibinfo{person}{Leslie~G. Valiant}.}
  \bibinfo{year}{1975}\natexlab{}.
\newblock \showarticletitle{General Context-Free Recognition in Less than Cubic
  Time}.
\newblock \bibinfo{journal}{\emph{J. Comput. Syst. Sci.}} \bibinfo{volume}{10},
  \bibinfo{number}{2} (\bibinfo{year}{1975}), \bibinfo{pages}{308--315}.
\newblock
\urldef\tempurl%
\url{https://doi.org/10.1016/S0022-0000(75)80046-8}
\showDOI{\tempurl}


\bibitem[\protect\citeauthoryear{Veldhuizen}{Veldhuizen}{2014}]%
        {DBLP:conf/icdt/Veldhuizen14}
\bibfield{author}{\bibinfo{person}{Todd~L. Veldhuizen}.}
  \bibinfo{year}{2014}\natexlab{}.
\newblock \showarticletitle{Triejoin: {A} Simple, Worst-Case Optimal Join
  Algorithm}. In \bibinfo{booktitle}{\emph{Proc. 17th International Conference
  on Database Theory ((ICDT) 2014)}}. \bibinfo{pages}{96--106}.
\newblock
\urldef\tempurl%
\url{https://doi.org/10.5441/002/icdt.2014.13}
\showDOI{\tempurl}


\bibitem[\protect\citeauthoryear{Williams}{Williams}{2015}]%
        {Williams15}
\bibfield{author}{\bibinfo{person}{Virginia~Vassilevska Williams}.}
  \bibinfo{year}{2015}\natexlab{}.
\newblock \showarticletitle{Hardness of Easy Problems: Basing Hardness on
  Popular Conjectures such as the Strong Exponential Time Hypothesis}. In
  \bibinfo{booktitle}{\emph{Proceedings of the 10th International Symposium on
  Parameterized and Exact Computation ({IPEC} 2015)}}
  \emph{(\bibinfo{series}{LIPIcs})}, Vol.~\bibinfo{volume}{43}.
  \bibinfo{pages}{17--29}.
\newblock


\bibitem[\protect\citeauthoryear{Zhuk}{Zhuk}{2017}]%
        {DBLP:conf/focs/Zhuk17}
\bibfield{author}{\bibinfo{person}{Dmitriy Zhuk}.}
  \bibinfo{year}{2017}\natexlab{}.
\newblock \showarticletitle{A Proof of {CSP} Dichotomy Conjecture}. In
  \bibinfo{booktitle}{\emph{58th {IEEE} Annual Symposium on Foundations of
  Computer Science (FOCS 2017), Berkeley, CA, USA, October 15-17, 2017}}.
  \bibinfo{pages}{331--342}.
\newblock
\urldef\tempurl%
\url{https://doi.org/10.1109/FOCS.2017.38}
\showDOI{\tempurl}


\end{thebibliography}

\end{document}